\documentclass[14pt]{extarticle}
\usepackage{extsizes}
\usepackage{amssymb,amsfonts,amsthm}
\usepackage{amsmath}
\usepackage[makeroom]{cancel}
\usepackage{mathtools}
\usepackage{mathrsfs,pifont}
\usepackage{slashed,mathabx}
\usepackage[bbgreekl]{mathbbol}
\usepackage{ytableau}
\usepackage{tikz}
\usetikzlibrary{calc}

\usepackage[top=1in, bottom=1.25in, left=0.75in,right=0.75in]{geometry}
\usepackage{authblk}

\usepackage[font={small,sl}]{caption}
\usepackage[all]{xy}

\usepackage{hyperref}

\newtheorem{theorem}{Theorem}[section]
\newtheorem{lemma}[theorem]{Lemma}
\newtheorem{corollary}[theorem]{Corollary}
\theoremstyle{definition}
\newtheorem{definition}[theorem]{Definition}
\newtheorem{proposition}[theorem]{Proposition}

\theoremstyle{remark}
\newtheorem{remark}[theorem]{Remark}

\numberwithin{equation}{section}

\makeatletter
\def\blfootnote{\xdef\@thefnmark{}\@footnotetext}

\begin{document}

\title{Vertex operators, infinite wedge representations, \\ and correlation functions of  the $t$-Schur measure }

\author{Gary Greaves\thanks{Division of Mathematical Sciences, Nanyang Technological University, 21 Nanyang Link, Singapore 637371. Email: gary@ntu.edu.sg}
\qquad Naihuan Jing\thanks{Department of Mathematics, North Carolina State University, Raleigh, NC 27695, USA. Email: jing@ncsu.edu}
\qquad Haoran Zhu\thanks{Division of Mathematical Sciences, Nanyang Technological University, 21 Nanyang Link, Singapore 637371. Email: zhuh0031@e.ntu.edu.sg} }

\maketitle
	
\begin{abstract}
We study the $t$-Schur measure on partitions, defined by
$
\mathbb{P}(\lambda)=Z^{-1}S_\lambda(x;t)s_\lambda(y)
$,
where $S_\lambda(x;t)$ denotes the $t$-Schur symmetric functions and $s_\lambda(y)$ the ordinary Schur functions, and $Z$ is the normalising constant. Using vertex operator calculus, we realise $S_\lambda(x;t)$ in the charged free-fermion Fock space, yielding a $t$-deformation of the classical boson-fermion correspondence. These realisations give vertex-algebraic proofs of the $t$-Cauchy identities and $t$-Gessel identity. Building on this framework, we compute the correlation functions of the $t$-Schur measure and show that the associated point process is determinantal, with an explicit correlation kernel. The Poissonised $t$-Plancherel measure appears as a specialisation of our construction, so its correlation functions follow as a corollary. As an application, we derive the limiting distribution for the length of the longest ascent pair in a random permutation. Our results interpolate the Schur case at $t=0$, connect to the Schur-$Q$ theory at $t=-1$, and provide a probabilistic interpretation of a natural $t$-refinement of increasing subsequences via a generalised RSK correspondence.

\end{abstract}
	

	
\section{Introduction}

\subsection{A probabilistic model}\label{subsect:prob-model}
We begin with a probabilistic model that interpolates the Schur and shifted Schur (Schur's $Q$) worlds via a deformation parameter $t\le0$.

Fix integers $m,n\ge1$ and parameters $x=(x_1,\dots,x_m)$, $y=(y_1,\dots,y_n)$ with $0\le x_i,y_j\le1$ and $x_i y_j<1$.
Let $\mathcal A:=\{1'<1<2'<2<\cdots\}$ and denote by $\mathcal A_{m,n}$ the set of $m\times n$ matrices with entries in $\mathcal A\cup\{0\}$.
We equip the entries $\{a_{ij}\}$ with independent sitewise distributions parametrised by $x_i y_j$ and $t\le0$.

For $k\ge1$,
\begin{align*}
\mathbb P_t(a_{ij}=0)&=\frac{1-x_i y_j}{1-tx_i y_j},\\
\mathbb P_t(a_{ij}=k)&=\frac{1-x_i y_j}{1-tx_i y_j}\,(x_i y_j)^k,\\
\mathbb P_t(a_{ij}=k')&=\frac{1-x_i y_j}{1-tx_i y_j}\,(-t)\,(x_i y_j)^k,
\end{align*}
which indeed defines a probability mass function since
\[
\underbrace{\frac{1-x_i y_j}{1-tx_i y_j}}_{\mathbb P_t(a_{ij}=0)}
+\underbrace{\frac{1-x_i y_j}{1-tx_i y_j}\sum_{k\ge1}(x_i y_j)^{k}}_{\sum_{k\ge1}\mathbb P_t(a_{ij}=k)}
+\underbrace{\frac{1-x_i y_j}{1-tx_i y_j}\sum_{k\ge1}(-t)\,(x_i y_j)^{k}}_{\sum_{k\ge1}\mathbb P_t(a_{ij}=k')}
=1
\]
and all terms are nonnegative when $t\le0$.

For a matrix $A=(a_{ij})$, let $\mathrm{mark}(A)$ be the number of marked entries.
Define the column and row sums
\[
s_j=\sum_{i=1}^m |a_{ij}|\ (1\le j\le n),\qquad
u_i=\sum_{j=1}^n |a_{ij}|\ (1\le i\le m),
\]
and, for $s\in\mathbb Z_{\ge0}^{\,n}$, $u\in\mathbb Z_{\ge0}^{\,m}$ and $r\in\mathbb Z_{\ge0}$, put
\[
\mathcal A_{m,n;s,u,r}
=\Bigl\{A\in\mathcal A_{m,n}:\ \sum_{i=1}^m |a_{ij}|=s_j,\ \sum_{j=1}^n |a_{ij}|=u_i,\ \mathrm{mark}(A)=r\Bigr\}.
\]
Let
\[
Z_t(x,y)=\prod_{i=1}^m\prod_{j=1}^n\frac{1-tx_i y_j}{1-x_i y_j}
\]
be the Cauchy-type normalising constant.
By independence we obtain, for any $A\in\mathcal A_{m,n;s,u,r}$,
\begin{equation}\label{eq:pmass-new}
\mathbb P_t(\{A\})
=\Biggl(\prod_{i=1}^m\prod_{j=1}^n\frac{1-x_i y_j}{1-tx_i y_j}\Biggr)
(-t)^r \prod_{i=1}^m x_i^{u_i}\prod_{j=1}^n y_j^{s_j}
=\frac{1}{Z_t(x,y)}\,(-t)^r\,x^{u}\,y^{s}.
\end{equation}

Two natural questions arise:
\begin{enumerate}
\item[(i)] What is the pushforward distribution on Young diagrams under a suitable generalised RSK correspondence? \label{que-1}
\item[(ii)] What are the edge fluctuations (e.g.\ of $\lambda_1$) when $m,n\to\infty$ under canonical specialisations?
\end{enumerate}

\subsection{Symmetric functions and their $t$-analogue}\label{subsect:sym-fcn}

Let $\Lambda=\Lambda_{\mathbb{Q}}$ denote the ring of symmetric functions in the countably infinite set of variables $x_1,x_2,\ldots$ with coefficients in the field of rational numbers $\mathbb{Q}$.  The degree of a homogeneous symmetric function equips $\Lambda_{\mathbb{Q}}$ with the natural $\mathbb{Z}_{\geqslant 0}$-grading
\[
  \Lambda_{\mathbb{Q}} \;=\; \bigoplus_{k\geqslant 0} \Lambda_{\mathbb{Q}}^{k},
\]
where $\Lambda_{\mathbb{Q}}^{k}$ is the subspace consisting of all homogeneous symmetric functions of degree~$k$.

If the sequence of parts $(\lambda_i)$ is not required to be weakly decreasing, we call $\lambda$ a composition and still write $|\lambda|$ for its weight.  Whenever $\lambda$ is weakly decreasing and $|\lambda| = n$, we say that $\lambda$ is a partition of~$n$.
Endow $\Lambda_{\mathbb{Q}}$ with the standard inner product defined by the orthogonality of the power sums:
\[
  \bigl\langle p_\lambda,\,p_\mu \bigr\rangle
  \;=\;
  z_\lambda\,\delta_{\lambda\mu},
  \qquad
  z_\lambda \;=\; \prod_{i\ge1} i^{\,m_i} m_i!,
  \quad
  \text{for } \lambda = 1^{m_1} 2^{m_2}\dotsm .
\]
The Schur basis $\{s_\lambda\}_{\lambda\in\mathbb Y}$ is orthonormal with respect to this form,
where $\mathbb Y$ denotes the set of partitions (Young diagrams).

Define the generating series
\[
E_{x,t}(z)=\prod_{i\ge1}\frac{1+x_i z}{1+t x_i z}=\sum_{n\ge0}e^{(t)}_n(x)z^n,\quad
H_{x,t}(z)=\prod_{i\ge1}\frac{1-t x_i z}{1-x_i z}=\sum_{n\ge0}h^{(t)}_n(x)z^n,
\]
so that $E_{x,t}(z)H_{x,t}(-z)=1$.
For $\lambda\in\mathbb Y$ with conjugate $\lambda'$, set the \textbf{$t$-Schur function} via the Jacobi-Trudi-like identity (see in Macdonald's book~\cite[Chapter I, \S 5]{Ma})
\[
S_\lambda(x;t)=\det\bigl(h^{(t)}_{\lambda_i-i+j}(x)\bigr)
=\det\bigl(e^{(t)}_{\lambda_i'-i+j}(x)\bigr),
\]
which specialises to $S_\lambda(x;0)=s_\lambda(x)$.
The corresponding Cauchy identity reads
\[
\sum_{\lambda\in\mathbb Y} S_\lambda(x;t)\,s_\lambda(y)
=\prod_{i,j\ge1}\frac{1-t x_i y_j}{1-x_i y_j}.
\]
This is actually the \textbf{Hall-Littlewood kernel}. At $t=-1$ one recovers the kernel appearing in the Schur $Q$-theory and
connections to the shifted Schur world via the shifted RSK, though the measures are not identical in general.

\subsection{Generalised RSK over the alphabet $\mathcal A$}

Now we consider a combinatorial interpretation of the probabilistic model defined above. We work with the totally ordered \textbf{marked alphabet}
\[
\mathcal A=\{\,1'<1<2'<2<3'<3<\cdots\,\}.
\]
A \textbf{marked tableau} $S$ of shape $\lambda\in\mathbb Y$ is a filling of the Young diagram
$\lambda$ by letters in $\mathcal A$ such that:

\medskip
\noindent\textbf{(T1)} entries are weakly increasing along each row and down each column;

\noindent\textbf{(T2)} for each $k\ge1$, every row contains at most one marked $k'$ and every column contains at most one unmarked $k$.

\medskip
For a marked tableau $S$ let $\mathrm{mark}(S)$ denote the number of marked entries, and define
\[
x^{S}\ :=\ \prod_{i\ge1} x_i^{\,m_i(S)},\qquad m_i(S):=\#\{\text{letters }|i|\text{ (either }i'\text{ or }i)\text{ in }S\}.
\]
Likewise, for a (usual) semistandard tableau $U$ with entries in $\{1,2,\dots\}$ set
\[
y^{U}\ :=\ \prod_{j\ge1} y_j^{\,n_j(U)},\qquad n_j(U):=\#\{\text{letters }j\text{ in }U\}.
\]

Let $A=(a_{ij})\in\mathcal A_{m,n}$. For each $(i,j)$ with $|a_{ij}|=k>0$,
create $k$ copies of the pair
\[
(\beta,\alpha)=
\begin{cases}
(j,\ i') & \text{if } a_{ij}\text{ is marked},\\
(j,\ i)  & \text{if } a_{ij}\text{ is unmarked}.
\end{cases}
\]
Arrange all pairs in a \textbf{biword} $w_A=\binom{\beta_1\cdots\beta_N}{\alpha_1\cdots\alpha_N}$
by sorting primarily by $\beta$ in nondecreasing order, and for ties $\beta=j$
by increasing order of $\alpha$ in the alphabet $\mathcal A$:
\[
1' < 1 < 2' < 2 < \cdots .
\]
This \textit{lexicographic order on $(\beta,\alpha)$ with $\beta$ primary and $\alpha$ secondary} order is the natural analogue
of the classical matrix-to-biword encoding in RSK, and will guarantee the recording tableau is semistandard.

Starting from $(S_0,U_0)=(\varnothing,\varnothing)$ and reading the pairs of $w_A$ from left to right,
we insert $\alpha_k$ into the current marked tableau $S_{k-1}$ by the following row bumping rule:
\begin{itemize}
\item[\(\circ\)] if $\alpha_k$ is \emph{unmarked} ($\alpha_k\in\{1,2,\dots\}$), then in the current row
replace the leftmost entry $\gamma$ with $\gamma>\alpha_k$ (if any) by $\alpha_k$ and bump $\gamma$ to the next row; if none exists, append $\alpha_k$ to the end of the row and stop;
\item[\(\circ\)] if $\alpha_k$ is \emph{marked} ($\alpha_k\in\{1',2',\dots\}$), use the same procedure but with the weak inequality $\gamma\ge\alpha_k$.
\end{itemize}
Let $S_k$ be the result of inserting $\alpha_k$. Record $\beta_k$ in the new cell created in $S_k$
to obtain $U_k$ from $U_{k-1}$.
By construction, $U_k$ is a semistandard tableau, 
and the marked tableau $S_k$ satisfies \textbf{(T1)}\textbf{(T2)}: the switch from $>$ to $\ge$ for marked letters
is exactly what enforces \textbf{(T2)}.

\begin{theorem}[{\cite[Generalised RSK for $\mathcal A$]{Mat1}}]\label{thm:gen-rsk-A}
The above procedure is a bijection between matrices $A=(a_{ij})\in\mathcal A_{m,n}$
and pairs $(S,U)$ consisting of a marked tableau $S$ and a semistandard tableau $U$
of the same shape $\lambda\in\mathbb Y$. Moreover,
\[
\sum_{i=1}^m |a_{ij}|=\#\{\text{entries }j\text{ in }U\}=:s_j,\qquad
\sum_{j=1}^n |a_{ij}|=\#\{\text{entries }i\text{ in }S\}=:u_i,
\]
and $\mathrm{mark}(S)=\mathrm{mark}(A)$.
\end{theorem}

\begin{remark}
By the encoding above, each nonzero entry $a_{ij}$ contributes the weight $(x_i y_j)^{|a_{ij}|}$,
and if it is marked it contributes an extra factor $(-t)^{\mathbf 1_{\text{marked}}}$ in our probabilistic model.
Grouping by shape and summing over $S$ and $U$ of the same shape yields
\[
\sum_{\mathrm{shape}(S)=\lambda}(-t)^{\mathrm{mark}(S)}\,x^S=S_\lambda(x;t),\quad \text{and}\quad
\sum_{\mathrm{shape}(U)=\lambda}y^U = s_\lambda(y),
\]
where $S_\lambda(x;t)$ is the $t$-Schur polynomial and $s_\lambda(y)$ the Schur polynomial.
Thus the pushforward of the product measure to shapes is governed by the \emph{mixed} Cauchy kernel $\sum_\lambda S_\lambda(x;t)s_\lambda(y)$.
\end{remark}

For the lower row $\alpha_1\alpha_2\cdots\alpha_N$ of $w_A$, call a subsequence increasing
if it is weakly increasing in $\mathcal A$ and, for each $k\ge1$, it uses at most one marked letter $k'$.
Let $\ell(w_A)$ be the maximal length of such subsequences.

\begin{lemma}[{\cite[Lemma 1.1]{Mat1}}]\label{lem:column-length-A}
Suppose during the insertion of $\alpha_k$ into the first row of $S_{k-1}$ the new cell appears in column $j$.
Then every increasing subsequence of $\alpha_1\cdots\alpha_k$ ending at $\alpha_k$ has length at most $j$,
and there exists one with length exactly $j$.
\end{lemma}

\begin{theorem}[{\cite[Theorem 3]{Mat1}}]\label{thm:LIS-A}
Let $(S,U)$ be the image of $A$ under the generalised RSK. If $\mathrm{shape}(S)=\mathrm{shape}(U)=\lambda$,
then
\[
\ell(w_A)=\lambda_1.
\]
\end{theorem}

Now we can provide an answer to the first question posed in subsection~\ref{subsect:prob-model}:
For $A\in\mathcal A_{m,n}$ with column-sum vector $s=(s_1,\dots,s_n)$, row-sum vector $u=(u_1,\dots,u_m)$
and $r=\mathrm{mark}(A)$, independence gives
\begin{equation}\label{eq:pmass-final}
\mathbb P_t(\{A\})=\frac{1}{Z_t(x,y)}\,(-t)^r\,x^{u}\,y^{s},\qquad
x^{u}:=\prod_{i=1}^m x_i^{u_i},\ \ y^{s}:=\prod_{j=1}^n y_j^{s_j}.
\end{equation}
Summing over all matrices with $\ell(w_A)\le h$ and then pushing forward through the bijection
$(A\leftrightarrow (S,U))$ yields
\begin{align*}
\mathbb P_{t,m,n}(\ell\le h)
&= \mathbb P_t\bigl(\{\,A\in\mathcal A_{m,n} \,:\, \ell(w_A)\le h\,\}\bigr)\\
&= \sum_{s,u,r}\ \sum_{\substack{A\in\mathcal A_{m,n;s,u,r}\\ \ell(w_A)\le h}}
   \frac{1}{Z_t(x,y)}\,(-t)^r\,x^{u}\,y^{s}\\
&= \frac{1}{Z_t(x,y)}\sum_{\substack{\lambda\in\mathbb Y\\ \lambda_1\le h}}
   \Biggl[\sum_{\substack{S\ \mathrm{marked}\\ \mathrm{shape}(S)=\lambda}} (-t)^{\mathrm{mark}(S)}\,x^{S}\Biggr]
   \Biggl[\sum_{\substack{U\ \mathrm{SSYT}\\ \mathrm{shape}(U)=\lambda}} y^{U}\Biggr]\\
&= \frac{1}{Z_t(x,y)}\sum_{\substack{\lambda\in\mathbb Y\\ \lambda_1\le h}}
   S_\lambda(x_1,\dots,x_m;\,t)\,s_\lambda(y_1,\dots,y_n).
\end{align*}

\begin{remark}
The answer to Question~\ref{que-1} is not unique. The correspondence~\cite{Mat1} here is tailored to the $t$-Schur polynomials $S_\lambda(x;t)$ paired with $s_\lambda(y)$.
If one replaces the present insertion by the \textbf{shifted RSK} (adapted to strict partitions and Schur's $Q/P$-functions) with $t=-1$,
the induced pushforward law becomes the \textbf{shifted Schur measure} introduced by Tracy-Widom~\cite{TW}.

At $t=-1$ our Cauchy kernel coincides with that of the $Q/P$-theory. However, the measures are, in general, not identical, and they nevertheless share the same edge asymptotics under standard specialisations.
\end{remark}

\subsection{The $t$-Schur measure and organisation of the paper}

Building on the material above, we define the \textbf{$t$-Schur measure} on partitions, introduced by Matsumoto~\cite{Mat1}.

\begin{definition}\label{def:tSchur-measure}
Let $x=(x_i)_{i\ge1}$ and $y=(y_j)_{j\ge1}$ be specialisations of the power sums for which
\[
Z_t(x,y):=\sum_{\mu\in\mathbb Y} S_\mu(x;t)\,s_\mu(y)
\]
converges. Equivalently, by the mixed Cauchy identity,
\[
Z_t(x,y)=\prod_{i,j\ge1}\frac{1-t\,x_i y_j}{1-x_i y_j}.
\]
The \textbf{$t$-Schur measure} is
\begin{equation}\label{eq:tSchur-measure}
\mathbb P_t(\{\lambda\})
:=\frac{1}{Z_t(x,y)}\,S_\lambda(x;t)\,s_\lambda(y)\qquad(\lambda\in\mathbb Y).
\end{equation}
\end{definition}

\begin{remark}
At $t=0$ we recover the Schur measure of Okounkov~\cite{O}, that is,
$\mathbb P_0(\lambda)\propto s_\lambda(x)s_\lambda(y)$.
\end{remark}

It is shown that the $t$-Schur measure inherits much of the integrable structure of the Schur measure (see Section~\ref{sect:corr-fcn} later). In particular, its length-truncated sums admit Gessel-type Toeplitz determinants with symbol $H_{x,t}(z)H_y(z^{-1})$, and hence Fredholm determinant representations via Borodin-Okounkov factorisation~\cite{BO}. Under the $\alpha$-specialisation, Matsumoto~\cite{Mat1} showed that the top row $\lambda_1$ exhibits Tracy-Widom $F_2$ fluctuations in the standard soft-edge scaling, as in the (shifted) Schur case~\cite{Joh2,TW}. More broadly, random partitions under Schur-type measures form determinantal point processes whose
edge scaling belongs to the same universality class as unitary random matrices.

These features place the $t$-Schur framework alongside a family of integrable deformations and extensions:
the periodic Schur process~\cite{BoPer}, shifted Schur measure~\cite{TW,Mat2}, Macdonald processes~\cite{BC}, and Hall-Littlewood models~\cite{BufPet},
among others. However, we would like to stress that the current results are different from the process arising from Hall-Littlewood measures.

The paper is organised as follows. In Section~\ref{sect:VO-realise}, we give a vertex operator realisation of the $t$-Schur functions.  As byproducts, we show vertex-algebraic proofs of the $t$-deformed Cauchy identity and the Gessel-type determinant for $\sum_{\ell(\lambda)\leqslant k}S_\lambda(x;t)s_\lambda(y)$, with Toeplitz symbol $\phi(z)=H_{x,t}(z)H_y(z^{-1})$. In Section~\ref{sect:infinite-wedge}, we recall the infinite wedge representation and establish the corresponding boson-fermion correspondence relevant to~\eqref{eq:tSchur-measure}. In Section~\ref{sect:corr-fcn}, using vertex operator calculus, we compute correlation functions and show that the $t$-Schur measure is determinantal, giving explicit kernels and their Fredholm representations. In Sections~\ref{sect:t-Plancherel-measure}, \ref{sect:t-ascent-pair}, and \ref{sect:limit-distri}, we give some applications: a $t$-Plancherel (Poissonised) family and a $t$-version of ascent pairs interpolating between longest increasing subsequence ($t=0$) and the shifted case ($t=-1$).
We also obtain scaling limits for $\lambda_1$ via correlation kernels, where Matsumoto's main result~\cite[Theorem~1]{Mat1} appears as a specialisation.
For $t=0$, our results recover those of Johansson~\cite{Joh2}, Okounkov~\cite{O}, and Borodin-Okounkov-Olshanski~\cite{BOO}.

\section{Vertex operators realisations}\label{sect:VO-realise}
Classical symmetric functions occupy a central place across mathematics and physics. In particular, they are closely connected to representations of Lie algebras in (in)finite dimensions~\cite{W}, integrable systems, and probability measures~\cite{Mat2,O}.

Building on Bernstein's seminal work~\cite{Ze}, vertex operator techniques have provided systematic constructions for prominent families of symmetric functions, including the Schur and Schur $Q$-functions~\cite{J1}, as well as the Hall-Littlewood functions~\cite{J2}, and some particular cases of Jack functions~\cite{CJ} and Macdonald functions~\cite{J3, J4, Za}.

\subsection{Vertex operators arising from $t$-Schur functions}

Begin by considering the ring $\Lambda_{\mathbb{Q}}$ as a Fock space associated with the infinite-dimensional Heisenberg algebra.  Define
\[
  a_{-n}=p_{n}\quad(n\geqslant 1),
  \qquad
  a_{n}=n\frac{\partial}{\partial p_{n}},
\]
where the power-sum symmetric function is $p_{n}(x)=\sum_{i=1}^{\infty} x_i^{\,n}$.
The operators $\{a_{n}\mid n\neq0\}$ together with the central element
$c=I$ generate a subalgebra
\(
  \mathcal{H}\subset\operatorname{End}(\Lambda)
\)
that satisfies
\[
  [a_{m},a_{n}]=m\delta_{m,-n}\,c,
  \qquad
  [c,a_{n}]=0,
\]
and is therefore isomorphic to the (complex) infinite-dimensional
Heisenberg algebra.  The space $\Lambda$ is the unique irreducible
$\mathcal{H}$-module determined by the conditions
\[
  a_{n}.1=0\quad(n>0),
  \qquad
  c=1.
\]
Equip $\Lambda$ with the Hermitian form for which
$a_{n}^{\ast}=a_{-n}$.  The monomial basis
$a_{-\lambda}=p_{\lambda}$ is orthogonal, and
\[
  \langle a_{-\lambda},a_{-\mu}\rangle
  =z_{\lambda}\,\delta_{\lambda\mu}.
\]

Now let us recall the vertex operator realisation of Schur functions.
Define the \textbf{Bernstein vertex operators}~\cite{Ze} (resp. \textbf{Jing operators}~\cite{J1}\footnote{Jing operators are the adjoint operators of Bernstein vertex operators and adjoint is taken with respect to the standard scalar product on $\Lambda$.}) $S(z)$ (resp. $S^{*}(z)$) $: \Lambda \longrightarrow \Lambda[[z,z^{-1}]]$ as linear maps,
by
\[
\begin{aligned}
S(z)
  &=\exp\Bigl(\sum_{n=1}^{\infty}\frac{a_{-n}}{n}\,z^{n}\Bigr)\,
    \exp\Bigl(-\sum_{n=1}^{\infty}\frac{a_{n}}{n}\,z^{-n}\Bigr) \\
  &=\exp\Bigl(\sum_{n=1}^{\infty}\frac{p_{n}}{n}\,z^{n}\Bigr)\,
    \exp\Bigl(-\sum_{n=1}^{\infty}\frac{\partial}{\partial p_{n}}\,
                z^{-n}\Bigr)
   \;=\sum_{n\in\mathbb{Z}}S_{n}\,z^{-n},
\\
S^{*}(z)
  &=\exp\Bigl(-\sum_{n=1}^{\infty}\frac{a_{-n}}{n}\,z^{n}\Bigr)\,
    \exp\Bigl(\sum_{n=1}^{\infty}\frac{a_{n}}{n}\,z^{-n}\Bigr) \\
  &=\exp\Bigl(-\sum_{n=1}^{\infty}\frac{p_{n}}{n}\,z^{n}\Bigr)\,
    \exp\Bigl(\sum_{n=1}^{\infty}\frac{\partial}{\partial p_{n}}\,
                z^{-n}\Bigr)
   \;=\sum_{n\in\mathbb{Z}}S^{*}_{n}\,z^{n}.
\end{aligned}
\]

\begin{proposition}[\cite{J1}] \label{prop:OPE}
The operator product expansions satisfy
\begin{align}
S(z)\,S(w)   &=\,:S(z)S(w):\,(1-wz^{-1}),\\
S^{*}(z)\,S^{*}(w) &=\,:S^{*}(z)S^{*}(w):\,(1-wz^{-1}),\\
S(z)\,S^{*}(w) &=\,:S(z)S^{*}(w):\,(1-wz^{-1})^
{-1},
\end{align}
where \(|w|<\min\{|z|,|z|^{-1}\}\). The rational factors are understood as formal power-series expansions in the variable \(w\).
\end{proposition}

We now pass to the $t$-deformed setting.
The vertex operators for the $t$-Schur functions are defined as $Y(z),\,Y^{*}(z):\Lambda \longrightarrow \Lambda[[z,z^{-1}]]$,
with
\begin{align}
Y(z)
   &=\exp\Bigl(\sum_{n=1}^{\infty}\frac{1-t^{n}}{n}\,a_{-n}\,z^{n}\Bigr)\,
     \exp\Bigl(-\sum_{n=1}^{\infty}\frac{1}{1-t^{n}}\frac{a_{n}}{n}\,z^{-n}\Bigr) \\\nonumber
   &=\exp\Bigl(\sum_{n=1}^{\infty}\frac{1-t^{n}}{n}\,p_{n}\,z^{n}\Bigr)\,
     \exp\Bigl(-\sum_{n=1}^{\infty}\frac{1}{1-t^{n}}\frac{\partial}{\partial p_{n}}\,z^{-n}\Bigr)
    \;=\sum_{n\in\mathbb{Z}}Y_{n}\,z^{-n},
\\
Y^{*}(z)
   &=\exp\Bigl(-\sum_{n=1}^{\infty}\frac{1-t^{n}}{n}a_{-n}\,z^{n}\Bigr)\,
     \exp\Bigl(\sum_{n=1}^{\infty}\frac{1}{1-t^{n}}\frac{a_{n}}{n}\,z^{-n}\Bigr) \\
   &=\exp\Bigl(-\sum_{n=1}^{\infty}\frac{1-t^{n}}{n}\,p_{n}\,z^{n}\Bigr)\,
     \exp\Bigl(\sum_{n=1}^{\infty}\frac{1}{1-t^{n}}\frac{\partial}{\partial p_{n}}\,z^{-n}\Bigr)
    \;=\sum_{n\in\mathbb{Z}}Y_{n}^{*}\,z^{\,n}.\nonumber
\end{align}

\begin{remark}
    When $t=0$, we recover the Bernstein-Jing case, that is, $Y(z;t=0)=S(z)$ and $Y^*(z;t=0)=S^*(z)$.
\end{remark}

The \textbf{normal-ordered product}, written $:\cdots:$, is obtained by commuting every creation operator $a_{-n}$ to the left of every annihilation operator $a_{n}$. For the $t$-Schur vertex operators, this gives, for example,
\[
\begin{aligned}
: Y(z)Y(w) :
   &=
     \exp\Bigl(\sum_{n=1}^{\infty}
               \frac{1-t^{n}}{n}\,a_{-n}\,(z^{n}+w^{n})\Bigr)\,
     \exp\Bigl(-\sum_{n=1}^{\infty}
              \frac{1}{1-t^{n}}\frac{a_{n}}{n}\,(z^{-n}+w^{-n})\Bigr)\\
   &=
     \exp\Bigl(\sum_{n=1}^{\infty}
              \frac{1-t^{n}}{n}\,p_{n}\,(z^{n}+w^{n})\Bigr)\,
     \exp\Bigl(-\sum_{n=1}^{\infty}
               \frac{1}{1-t^{n}}\frac{\partial}{\partial p_{n}}\,(z^{-n}+w^{-n})\Bigr).
\end{aligned}
\]

\begin{proposition}
The vertex operators satisfy the following relations on $\Lambda$:
\begin{align}
Y(z)\,Y(w) &= :Y(z)\,Y(w):(1-wz^{-1}), \\
Y^*(z)\,Y^\ast(w) &= :Y^*(z)\,Y^\ast(w):(1-wz^{-1}), \\
Y(z)\,Y^*(w) &= :Y(z)\,Y^*(w):\,(1-wz^{-1})^{-1}.
\end{align}
where \(|w|<\min\{|z|,|z|^{-1}\}\).
\end{proposition}

\begin{proposition}
The operators $Y_n$ and $Y_n^\ast$ satisfy the Clifford relations:
\begin{align}
\{Y_m,Y_n^*\}=\delta_{m,n}I,\qquad \{Y_m,Y_n\}=\{Y_m^*,Y_n^*\}=0
\end{align}
Here $\delta_{m,n}$ denotes the Kronecker delta.
\end{proposition}

Although the $t$-Schur vertex operators $Y(z),Y^*(z)$ explicitly depend on $t$, in the Baker-Campbell-Hausdorff contraction, the creation and annihilation coefficients multiply to 1, so the OPE rational prefactors coincide with the Schur case.

For example, writing $Y(z)=e^{A(z)}e^{B(z)}$ with $A(z)=\sum_{n\geqslant 1}\frac{1-t^n}{n}\,a_{-n}z^n$ and $B(z)=-\sum_{n\geqslant 1}\frac{1}{1-t^n}\frac{a_n}{n}z^{-n}$, the only nonzero BCH contraction is
\[
[B(z),A(w)]
=-\sum_{n\ge1}\frac{1}{1-t^n}\frac{1}{n}z^{-n}\cdot\frac{1-t^n}{n}w^n\,[a_n,a_{-n}]
=-\sum_{n\ge1}\frac{1}{n}\Bigl(\frac{w}{z}\Bigr)^{n}
=\log\bigl(1-wz^{-1}\bigr).
\]

\subsection{Realisations of the $t$-Schur function}\label{subsect:vo-realisation}

Expanding the Jacobi-Trudi determinant of $t$-Schur functions tells us
\[
S_{\lambda}(x;t)
=
\det
\begin{bmatrix}
h_{\lambda_1}^{(t)} & h_{\lambda_1+1}^{(t)} & \cdots & h_{\lambda_1+k-1}^{(t)}\\
h_{\lambda_2-1}^{(t)} & h_{\lambda_2}^{(t)} & \cdots & h_{\lambda_2+k-2}^{(t)}\\
\vdots & \vdots & \ddots & \vdots \\
h_{\lambda_k-k+1}^{(t)} & h_{\lambda_k-k+2}^{(t)} & \cdots & h_{\lambda_k}^{(t)}
\end{bmatrix},
\qquad k=\ell(\lambda),
\]
with the usual convention $h_m^{(t)}=0$ for $m<0$ and $h_0^{(t)}=1$.

Based on the preparations above, it is sufficient to give the vertex operator realisations of the $t$-Schur functions.

\begin{theorem}\label{thm:Jacobi-Trudi}
For any partition $\lambda=(\lambda_1,\dots,\lambda_k)$ with conjugate $\lambda'=(\lambda'_1,\lambda'_2,\dots)$ one has
\begin{equation}\label{e:det-tSchur}
Y_{-\lambda_1}Y_{-\lambda_2}\cdots Y_{-\lambda_k}\,.1
   \;=\;
S_\lambda(x;t)
   \;=\;
\det\bigl(h^{(t)}_{\lambda_i-i+j}\bigr)_{1\leqslant i,j\leqslant \ell(\lambda)},
\end{equation}
where $h^{(t)}_m$ denotes the $m$th $t$-complete symmetric function defined above.
\end{theorem}

\begin{proof}
Recall the generating series
\[
  H^{(t)}(z):=\exp\Bigl(\sum_{n\geqslant 1}\frac{1-t^n}{n}\,p_n z^n\Bigr)
  \;=\;\sum_{r\geqslant 0}h^{(t)}_r\,z^r,\qquad h^{(t)}_0=1,\; h^{(t)}_{r<0}=0.
\]
where $h_n^{(t)}$ is the generalised complete symmetric function.

By BCH and normal ordering, we have
\begin{equation}\label{eq:Wick}
  Y(z_1)\cdots Y(z_k).1 =
  \prod_{1\leqslant i<j\leqslant k}\bigl(1-\frac{z_j}{z_i}\bigr)
  :Y(z_1)\cdots Y(z_k):.1
  =
  \prod_{i<j}\bigl(1-\frac{z_j}{z_i}\bigr)
  \prod_{i=1}^k H^{(t)}(z_i).
\end{equation}


Recall a classical Vandermonde-type identity
\begin{equation}\label{eq:Vdm}
  \prod_{i<j}(1-\frac{z_j}{z_i})
  \;=\;
  \frac{\det\bigl(z_i^{\,k-j}\bigr)_{1\leqslant i,j\leqslant k}}{\prod_{i=1}^k z_i^{\,k-i}}
  \;=\;
  \sum_{\sigma\in S_k}\operatorname{sgn}(\sigma)
  \prod_{i=1}^k z_i^{\,i-\sigma(i)}.
\end{equation}
Then by taking the coefficient of $z_1^{\lambda_1}z_2^{\lambda_2}\cdots z_k^{\lambda_k}$ we get the result.
\end{proof}

Similarly, for the $Y^*(z)$, we have the following determinant formula.
\begin{theorem}
Let $\lambda=(\lambda_1,\dots,\lambda_k)$ be a partition of length $k$. Then
\[
  Y^*_{\lambda_1}\cdots Y^*_{\lambda_k}.1
  \;=\;
  (-1)^{|\lambda|}\,\det\bigl(e^{(t)}_{\lambda_i-i+j}\bigr)_{1\leqslant i,j\leqslant \ell(\lambda)}.
\]
\end{theorem}

By comparing the coefficients of the two theorems above, we obtain the following.

\begin{corollary}\label{thm:conjugate}
Let $\lambda=(\lambda_1,\dots,\lambda_\ell)$ be a partition and
$\lambda'=(\lambda_1',\dots,\lambda_k')$ its conjugate.
Then
\[
  Y_{-\lambda_1}\cdots Y_{-\lambda_\ell}. 1
   \;=\;
  (-1)^{|\lambda|}
  Y^{*}_{\lambda_1'}\cdots Y^{*}_{\lambda_k'}. 1 .
\]
Consequently $\{Y_{-\lambda}.1\}_\lambda$ and $\{Y^{*}_{\lambda'}.1\}_\lambda$ both form bases of the ring of symmetric function~$\Lambda$.
\end{corollary}

\subsection{Generalised Cauchy identities and Gessel identity}

\begin{theorem}\label{thm:t-Cauchy-VO}
For two variable sets $x=(x_1,x_2,\dots)$ and $y=(y_1,y_2,\dots)$ with $|x_i\,y_j|<1$, the $t$-Schur functions satisfy
\begin{equation}\label{eq:t-Cauchy}
   \sum_{\lambda\in\mathbb Y} S_\lambda(x;t)\,s_\lambda(y)=
   \prod_{i,j\geqslant 1}\frac{1-t\,x_i y_j}{1-x_i y_j}.
\end{equation}
\end{theorem}

\begin{proof}
Let $S(w^{-1})$ be the Bernstein operator for ordinary Schur functions and $Y(z)$ our $t$-vertex operator.
By BCH and $[a_m,a_n]=m\delta_{m,-n}$,
\[
\Bigl[-\sum_{n\geqslant 1}\frac{a_n}{n}\,w^{n},\ \sum_{m\geqslant1}\frac{1-t^m}{m}\,a_{-m}z^m\Bigr]
=-\sum_{n\geqslant 1}\frac{1-t^n}{n}\,(zw)^{n}
=\log\frac{1-zw}{1-tzw},
\]
hence we have the operator product expansion
\[
S(w^{-1})\,Y(z)=:S(w^{-1})\,Y(z):\,\frac{1-zw}{\,1-tzw\,}.
\]
Iterating over variables $x=(x_1,x_2,\dots)$ and $y=(y_1,y_2,\dots)$ (with $|x_i y_j|<1$ so that all series multiply formally), we obtain that
\[
\bigl\langle 1 \, ,\, \prod_{i\geqslant1}Y(x_i)\ \prod_{j\geqslant 1}S(y_j^{-1})\,.\,1 \bigr\rangle
=\prod_{i,j\geqslant 1}\frac{1-tx_i y_j}{1-x_i y_j}.
\]
On the other hand, by Theorem~\ref{thm:Jacobi-Trudi} and the standard Bernstein realisation,
\[
\prod_{i\geqslant1}Y(x_i).\,1=\sum_{\lambda\in\mathbb Y} S_\lambda(x;t)\,S_{-\lambda}.\,1,
\qquad
\prod_{j\geqslant 1}S(y_j^{-1}).\,1=\sum_{\mu\in\mathbb Y} s_\mu(y)\,S_{-\mu}.\,1.
\]
Therefore
\[
\bigl\langle 1 \, ,\, \prod_{i\geqslant1}Y(x_i)\ \prod_{j\geqslant1}S(y_j^{-1})\,.\,1 \bigr\rangle
=\sum_{\lambda\in\mathbb Y} S_\lambda(x;t)\,s_\lambda(y),
\]
and comparing the two evaluations yields
\[
\sum_{\lambda\in\mathbb Y} S_\lambda(x;t)\,s_\lambda(y)
=\prod_{i,j\geqslant 1}\frac{1-tx_i y_j}{1-x_i y_j},
\]
which is the desired $t$-Cauchy identity.
\end{proof}

\begin{theorem}\label{thm:t-dual-Cauchy}
For two variable sets $x=(x_1,x_2,\dots)$ and $y=(y_1,y_2,\dots)$ such that
$|x_i y_j|<1$, one has
\begin{equation}\label{eq:t-Cauchy-dual}
      \sum_{\lambda\in\mathbb Y }
           (-1)^{|\lambda|}
           \,S_\lambda(x;t)\,s_{\lambda'}(y)
      \;=\;
      \prod_{i,j\geqslant 1}\frac{1-x_i y_j}{1-t\,x_i y_j}.
\end{equation}
\end{theorem}

\begin{proof}
This time we pair the Jing operator $S^*(w^{-1})$ with $Y(z)$. As above,
\[
\Bigl[\sum_{n\geqslant 1}\dfrac{a_n}{n}\,w^{n}\,,\,\sum_{m\geqslant1}\dfrac{1-t^m}{m}\,a_{-m}z^m\Bigr]
=\sum_{n\geqslant 1}\frac{1-t^n}{n}\,(zw)^{n}
=\log\frac{1-tzw}{1-zw},
\]
hence
\[
S^*(w^{-1})\,Y(z)
=:S^*(w^{-1})\,Y(z):\;\frac{1-tzw}{\,1-zw\,}.
\]
Taking vacuum expectations gives
\[
\bigl<1\, ,\; \prod_{i\geqslant 1}Y(x_i)\;\prod_{j\geqslant 1}S^*(y_j^{-1})\,.1 \,\bigr>
=\prod_{i,j\geqslant 1}\frac{1-x_i y_j}{1-tx_i y_j}.
\]
On the other hand, the mode expansions together with the dual Jacobi-Trudi formula for ordinary Schur functions yield
\[
\prod_{j\geqslant 1}S^*(y_j^{-1}).1
=\sum_{\lambda\in\mathbb Y}(-1)^{|\lambda|}s_{\lambda'}(y)\,S^*_{\lambda}.1.
\]
Consequently,
\[
\bigl<1, \,\,\prod_{i\geqslant 1}Y(x_i)\;\prod_{j\geqslant 1}S^*(y_j^{-1})\,.1 \,\bigr>
=\sum_{\lambda\in\mathbb Y}(-1)^{|\lambda|}S_\lambda(x;t)\,s_{\lambda'}(y).
\]
Comparing with the previous evaluation establishes
\[
\sum_{\lambda\in\mathbb Y}(-1)^{|\lambda|}\,S_\lambda(x;t)\,s_{\lambda'}(y)
=\prod_{i,j\geqslant 1}\frac{1-x_i y_j}{1-tx_i y_j},
\]
which is the dual $t$-Cauchy identity.
\end{proof}

We keep employing a vertex algebraic approach to prove a $t$-analogue of the Gessel identity~\cite{Ge,Mat1}, which plays a crucial role in the proof of the limit theorem for the (shifted) Schur measure~\cite{Joh2, Mat1, TW}. We give this for completeness, while it is not needed for the vertex-algebraic derivation of correlation kernels later.

Write the Toeplitz symbol $\phi(z)=\sum_{m\in\mathbb Z}\phi_m z^m$ and
$T_k(\phi)=(\phi_{j-i})_{1\le i,j\le k}$.
Introduce the half-vertex operator
\(
\Gamma_-(z):=\exp\big(-\sum_{n\ge1}\frac{a_n}{n}\,z^{n}\big).
\)

\begin{theorem}\label{thm:tGessel}
For every integer $k\geqslant 1$,
\[
\sum_{\ell(\lambda)\leqslant k} S_\lambda(x;t)\,s_\lambda(y)
\;=\;
\det\left( T_k\big(\phi\big) \right),
\qquad
\phi(z):=H_{x,t}(z)\,H_y(z^{-1}),
\]
where $H_{x,t}(z)=\prod_{i\geqslant 1}\frac{1-t x_i z}{1-x_i z}$ and
$H_y(z)=\prod_{j\geqslant 1}\frac{1}{1-y_j z}$.
\end{theorem}

\begin{proof}
 Fix $k\geqslant 1$ and set $Z=(z_1,\dots,z_k)$.
From the OPEs
\begin{align*}
&Y(z_a)Y(z_b)=:Y(z_a)Y(z_b):\,(1-z_b z_a^{-1})\quad(a<b),\\
&
Y(z)\,\Gamma_-(x)=:Y(z)\,\Gamma_-(x):\ H_{x,t}(z).
\end{align*}
with $H_{x,t}(z)=\prod_{i\ge1}\frac{1-t\,x_i z}{1-x_i z}
=\sum_{m\ge0}h^{(t)}_m(x)z^m$, we obtain
\begin{equation}\label{eq:A-pureVO}
\bigl\langle 1,\; Y(z_1)\cdots Y(z_k)\,\Gamma_-(x)\,.1 \bigr\rangle
=\Bigl(\prod_{1\le a<b\le k}(1-z_b z_a^{-1})\Bigr)\,
  \prod_{i=1}^k H_{x,t}(z_i).
\end{equation}

Expand $Y(z)=\sum_{n\in\mathbb Z}Y_n z^{-n}$ and let $\rho=(k-1,k-2,\dots,0)$.
By expanding the Vandermonde factor and extracting coefficients, we have the standard identity
\begin{equation}\label{eq:B-pureVO}
\bigl[z_1^{-\lambda_1-\rho_1}\cdots z_k^{-\lambda_k-\rho_k}\bigr]\,
\Bigl(\prod_{1\le a<b\le k}(1-z_b z_a^{-1})\Bigr)\,
Y(z_1)\cdots Y(z_k)
=\det\!\bigl(Y_{-\lambda_i+j-i}\bigr)_{1\le i,j\le k}.
\end{equation}
Applying \eqref{eq:B-pureVO} to the correlator and using multilinearity
(equivalently, Wick's theorem for the $Y$-modes), we obtain, for $\ell(\lambda)\le k$,
\[
\bigl\langle 1,\; \det\!\bigl(Y_{-\lambda_i+j-i}\bigr)_{i,j=1}^k\,\Gamma_-(x)\,.1 \bigr\rangle
=\det\!\bigl(\,\langle 1,\,Y_{-\lambda_i+j-i}\,\Gamma_-(x)\,.1\rangle\,\bigr)_{1\le i,j\le k}.
\]
Since
\[
\langle 1,\,Y(z)\,\Gamma_-(x)\,.1\rangle
=H_{x,t}(z)=\sum_{m\ge0}h_m^{(t)}(x)\,z^{m},
\]
we finally get
\[
\bigl\langle 1,\; \det\!\bigl(Y_{-\lambda_i+j-i}\bigr)\,\Gamma_-(x)\,.1 \bigr\rangle
=\det\bigl(h^{(t)}_{\lambda_i-i+j}(x)\bigr)_{1\le i,j\le k}
=S_\lambda(x;t).
\]

Moreover, multiplying the left-hand side of \eqref{eq:B-pureVO} by
$\prod_{i=1}^k H_y(z_i^{-1})$ and taking the same
coefficient gives
\[
\bigl[z_1^{-\lambda_1-\rho_1}\cdots z_k^{-\lambda_k-\rho_k}\bigr]\,
\Bigl(\prod_{1\le a<b\le k}(1-z_b z_a^{-1})\Bigr)\,
\prod_{i=1}^k H_y(z_i^{-1})
=\det\bigl(h_{\lambda_i-i+j}(y)\bigr)_{1\le i,j\le k}
=s_\lambda(y).
\]

Multiply \eqref{eq:A-pureVO} by $\prod_{i=1}^k H_y(z_i^{-1})$, extract
$(z_1^{-\lambda_1-\rho_1}\cdots z_k^{-\lambda_k-\rho_k})$, and sum over all
$\lambda$ with $\ell(\lambda)\le k$:
\[
\sum_{\ell(\lambda)\le k} S_\lambda(x;t)s_\lambda(y)
=\sum_{\ell(\lambda)\le k}
\bigl[z_1^{-\lambda_1-\rho_1}\cdots z_k^{-\lambda_k-\rho_k}\bigr]
\Bigl(\prod_{1\le a<b\le k}(1-z_b z_a^{-1})\Bigr)
\prod_{i=1}^k \bigl(H_{x,t}(z_i)H_y(z_i^{-1})\bigr).
\]

Now use the determinant form of the Vandermonde-type factor:
\[
\prod_{1\le a<b\le k}\bigl(1-z_b z_a^{-1}\bigr)
=\Bigl(\prod_{i=1}^k z_i^{\,k-i}\Bigr)^{-1}
\det\bigl(z_i^{\,j-i}\bigr)_{1\le i,j\le k}.
\]
Here, the additional factor \(\prod_i z_i^{\,k-i}\) is exactly offset by
\(\rho=(k-1,\dots,0)\) when we take the coefficient. Thus, we obtain the coefficient by taking the multiple linear coefficient of each \(z_i\).
\[
\sum_{\ell(\lambda)\le k} S_\lambda(x;t)\,s_\lambda(y)
=\det\Big([z^{\,j-i}]\,H_{x,t}(z)\,H_y(z^{-1})\Big)_{1\le i,j\le k}
=\det T_k\big(\phi\big),
\]
which proves the claim.
\end{proof}

\begin{remark}
Here, we use the \textit{length truncation} $\ell(\lambda)\le k$, whose Toeplitz symbol is $\phi(z)=H_{x,t}(z)\,H_y(z^{-1})$.  One may see another Gessel-type identity~\cite[Lemma 2]{Mat1}, depending on a different truncation, that is, for the \textit{first row truncation} $\lambda_1\le h$, the correct symbol changes to $ E_{x,t}(z^{-1})\,E_y(z)$.
\end{remark}

\section{Infinite wedge representations}\label{sect:infinite-wedge}

The infinite wedge (more precisely, the half-infinite wedge $\bigwedge^{\infty/2} V$, i.e.\ the fermionic Fock space) has connections to representation theory of infinite-dimensional symmetric groups, integrable systems, and modular forms, among many other areas. Here, we record only the features needed for the Boson-Fermion correspondence and relate the character theory of the symmetric group. Standard references include \cite{BlO,BO,Kac,O}. For
categorification of the Boson-Fermion correspondence and the infinite symmetric group, please see \cite{FPS}.

\subsection{Fermionic Fock space}

Let $V$ be a complex vector space with basis $\{e_k\}_{k\in\mathbb Z+\frac12}$.
The half-infinite wedge $\Lambda^{\infty/2}V$ is spanned by vectors
\[
  v_S \;=\; e_{s_1}\wedge e_{s_2}\wedge e_{s_3}\wedge\cdots,
\]
where $S=\{s_1>s_2>\cdots\}\subset\mathbb Z+\tfrac12$ satisfies
\[
  S_+\;=\;S\setminus\bigl(\mathbb Z_{\le0}-\tfrac12\bigr),\qquad
  S_-\;=\;\bigl(\mathbb Z_{\le0}-\tfrac12\bigr)\setminus S,
\]
both finite. We equip $\Lambda^{\infty/2}V$ with the inner product for which $\{v_S\}$ is orthonormal.
The charge-$m$ vacuum is
\[
  v_m\;=\;\bigl(m-\tfrac12\bigr)\wedge\bigl(m-\tfrac32\bigr)\wedge\bigl(m-\tfrac52\bigr)\wedge\cdots,
\]
and in particular $v_0$ is the Dirac sea.

\subsection{Free fermions}

Define operators $\psi$ and their adjoints $\psi^*$ on $\Lambda^{\infty/2}V$ by
\[
  \psi_k(f)=e_k\wedge f,\quad
  \psi_k^*(e_k\wedge f)=f,\quad \psi_k^*(e_\ell\wedge f)=0\,\, (\ell\neq k).
\]
They satisfy the canonical anticommutation relations (CAR)
\[
  \{\psi_k,\psi_\ell^{*}\}=\delta_{k\ell},\qquad
  \{\psi_k,\psi_\ell\}=0,\qquad
  \{\psi_k^{*},\psi_\ell^{*}\}=0,
\]
hence
\(
  \psi_k\psi_k^{*}\,v_S=v_S \ (k\in S),\ \psi_k\psi_k^{*}\,v_S=0 \ (k\notin S).
\)
Introduce generating fields (formal Laurent series)
\[
  \psi(z)=\sum_{k\in\mathbb Z+\tfrac12} z^{\,k}\,\psi_k,\qquad
  \psi^{*}(z)=\sum_{k\in\mathbb Z+\tfrac12} z^{-k}\,\psi^{*}_k,
\]
and the normal ordering (with respect to the Dirac sea)
\[
  :\psi_k\psi_k^{*}:\;=\;
  \begin{cases}
    \psi_k\psi_k^{*},& k>0,\\
    -\,\psi_k^{*}\psi_k,& k<0.
  \end{cases}
\]

\subsection{Bosons and $t$-half vertex operators}\label{subsec:t-half-vertex}

Set
\[
  \alpha_n \;=\;\sum_{k\in\mathbb Z+\tfrac12} :\psi_{k-n}\psi_k^{*}:,\qquad n\in\mathbb Z\setminus\{0\}.
\]
Then
\[
  [\alpha_n,\alpha_m]=n\,\delta_{n,-m},\qquad \alpha_n^*=\alpha_{-n},
\]
and
\begin{equation}\label{eq:Heis-on-psi}
  [\alpha_n,\psi(z)]=z^{\,n}\psi(z),\qquad
  [\alpha_n,\psi^{*}(z)]=-\,z^{\,n}\psi^{*}(z).
\end{equation}

According to the vertex operator realisation in Section~\ref{sect:VO-realise}, we consider the following
\textbf{half-vertex operators}\footnote{Under the Boson-Fermion correspondence (charge $0$),
creation modes are $a_{-n}$ and $\alpha_{-n}$, and annihilation modes are $a_n$ and $\alpha_n$ ($n>0$).
The identifications are $\alpha_{-n}\leftrightarrow a_{-n}=p_n$ and $\alpha_n\leftrightarrow a_n=n\,\partial/\partial p_n$ ($n\ge1$), so $[\alpha_n,\alpha_m]=n\delta_{n,-m}=[a_n,a_m]$.}
, given any sequence $u = (u_1, u_2, \ldots)$,
\begin{equation}\label{eq:Y-def}
  \mathfrak Y_{+}^{(t)}(u):=\exp\Bigl(\sum_{n\ge1}\frac{1-t^{n}}{n}\,\alpha_{n}\,u_n\Bigr),\qquad
  \mathfrak Y_{-}(u):=\exp\Bigl(\sum_{n\ge1}\frac{1}{n}\,\alpha_{-n}\,u_n\Bigr).
\end{equation}
By \eqref{eq:Heis-on-psi} and the Baker-Campbell-Hausdorff formula,
\begin{align}
  &\mathfrak Y_{+}^{(t)}(x)\,\psi(z)\,\bigl(\mathfrak Y_{+}^{(t)}(x)\bigr)^{-1}
   =H_{x,t}(z)\,\psi(z),\\
   &\mathfrak Y_{+}^{(t)}(x)\,\psi^{*}(z)\,\bigl(\mathfrak Y_{+}^{(t)}(x)\bigr)^{-1}
   =H_{x,t}(z)^{-1}\,\psi^{*}(z),\label{eq:tY-conj-plus}\\
  &\mathfrak Y_{-}(y)\,\psi(z)\,\bigl(\mathfrak Y_{-}(y)\bigr)^{-1}
   =H_{y}(z^{-1})\,\psi(z),\\
   &\mathfrak Y_{-}(y)\,\psi^{*}(z)\,\bigl(\mathfrak Y_{-}(y)\bigr)^{-1}
   =H_{y}(z^{-1})^{-1}\,\psi^{*}(z),\label{eq:tY-conj-minus}
\end{align}
where
\begin{equation}\label{eq:HxHy}
  H_{x,t}(z)=\exp\Bigl(\sum_{n\ge1}\frac{1-t^n}{n}\,p_n(x)\,z^n\Bigr)
            =\prod_{i}\frac{1-t\,x_i z}{1-x_i z},\quad
  H_{y}(z)
          =\prod_{j}\frac{1}{1-y_j z}.
\end{equation}
Moreover,
\begin{align}\label{eq:t-Cauchy-half}
  &\mathfrak Y_{+}^{(t)}(x)\,\mathfrak Y_{-}(y)
  \;=\;Z_t(x,y)\,\mathfrak Y_{-}(y)\,\mathfrak Y_{+}^{(t)}(x),\\
  &\log Z_t(x,y)=\sum_{n\ge1}\frac{1-t^n}{n}\,p_n(x)\,p_n(y)
  =\sum_{n\ge1}n\,(1-t^n)\,\mathfrak t_n \mathfrak t'_n,
\end{align}
with Miwa times $\mathfrak t_n=p_n(x)/n$, $\mathfrak t_n'=p_n(y)/n$.
Finally, since $\alpha_n v_m=0$ for $n>0$, we have
\[
  \mathfrak Y_{+}^{(t)}(x)\,v_m=v_m\qquad\text{for all }m\in\mathbb Z.
\]

\section{Correlation functions of $t$-Schur measure}\label{sect:corr-fcn}

\subsection{Coordinates on partitions}

To a partition $\lambda$, we associate the subset
\[
  \mathfrak S(\lambda)=\{\lambda_i-i+\tfrac12\}\subset\mathbb Z+\tfrac12.
\]
For example, $\mathfrak S(\varnothing)=\{-\tfrac12,-\tfrac32,-\tfrac52,\dots\}$.
One has $S=\mathfrak S(\lambda)$ if and only if
\[
  S_+ = S\setminus\bigl(\mathbb Z_{\le 0}-\tfrac12\bigr),\qquad
  S_- = \bigl(\mathbb Z_{\le 0}-\tfrac12\bigr)\setminus S
\]
are finite and $|S_+|=|S_-|$; this common value equals the number of diagonal boxes of~$\lambda$.
The finite set $S_+(\lambda)\cup S_-(\lambda)$ is the modified Frobenius coordinate set~\cite{KO}.

Given a finite $X\subset\mathbb Z+\tfrac12$ (we always list $X=\{x_1<\cdots<x_n\}$ in increasing order),
define the \textbf{$n$-point correlation function}
\[
  \mathfrak C_t(X)
  = \mathbb P_t^{x,y}\bigl(\{\lambda:\ X\subset \mathfrak S(\lambda)\}\bigr),
\]
where $\mathbb P_t^{x,y}(\{\lambda\})=Z_t(x,y)^{-1}S_\lambda(x;t)s_\lambda(y)$ is the
normalised $t$-Schur measure and
\(
Z_t(x,y)=\prod_{i,j}\frac{1-tx_i y_j}{1-x_i y_j}.
\)

\subsection{Determinantal structure}

Using half-vertex operators in subsection~\ref{subsec:t-half-vertex}, we have
\begin{equation}\label{eq:Ct-as-functional}
  \mathfrak C_t(X)
  = \frac{1}{Z_t(x,y)}\,
    \bigl\langle v_0\, ,\, \mathfrak Y_+^{(t)}(x)\,\bigl(\prod_{a\in X}\psi_a\psi_a^{*}\bigr)\,\mathfrak Y_-(y)\, v_0 \bigr\rangle.
\end{equation}
Set the Bogoliubov transform
\[
  \mathcal G_t := \mathfrak Y_+^{(t)}(x)\,\mathfrak Y_-(y)^{-1},\qquad
  \mathfrak F_k:=\mathcal G_t\,\psi_k\,\mathcal G_t^{-1},\qquad
  \mathfrak F_k^{*}:=\mathcal G_t\,\psi_k^{*}\,\mathcal G_t^{-1}.
\]
Using $\mathfrak Y_+^{(t)}\mathfrak Y_- = Z_t\,\mathfrak Y_-\,\mathfrak Y_+^{(t)}$ and the CAR, one rewrites the previous expectation as the vacuum expectation
\[
  \mathfrak C_t(X)=\bigl\langle v_0\, ,\, \prod_{x\in X}\mathfrak F_x\mathfrak F_x^{*}\, v_0 \bigr\rangle.
\]

\begin{theorem}\label{thm:main}
For any finite $X\subset\mathbb Z+\tfrac12$,
\begin{align}\label{eq:det-claim}
  \mathfrak C_t(X)=\det\bigl[\,\mathcal K_t(x_i,x_j)\,\bigr]_{x_i,x_j\in X},
\qquad
  \mathcal K_t(i,j) := \bigl\langle v_0\, ,\, \mathfrak F_i\,\mathfrak F_j^{*}\, v_0 \bigr\rangle.
\end{align}
\end{theorem}

\begin{proof}

By \eqref{eq:Ct-as-functional} and the convention that the product
$\prod_{a\in X}$ is taken in the increasing order of~$a$,
\begin{align*}
  \mathfrak C_t(X)
  &= \frac{1}{Z_t(x,y)}\,
    \bigl\langle v_0\, ,\, \mathfrak Y_+^{(t)}(x)\,
      \bigl(\psi_{u_1}\psi_{u_1}^*\cdots \psi_{u_n}\psi_{u_n}^*\bigr)\,
     \mathfrak Y_-(y) \, v_0 \bigr\rangle\\
  &=
  \bigl\langle v_0\, ,\, \mathfrak F_{u_1}\cdots\mathfrak F_{u_n}\,
                     \mathfrak F_{u_n}^{*}\cdots\mathfrak F_{u_1}^{*}\, v_0 \bigr\rangle,
\end{align*}
where in the last equality we used
$\mathcal G_t=\mathfrak Y_+^{(t)}(x)\mathfrak Y_-(y)^{-1}$ and
$\mathfrak Y_+^{(t)}\mathfrak Y_-=Z_t\,\mathfrak Y_-\,\mathfrak Y_+^{(t)}$.

The CAR imply, for $i\neq j$,
\[
  \mathfrak F_{u_i}\,\mathfrak F_{u_j}^{*}=-\,\mathfrak F_{u_j}^{*}\,\mathfrak F_{u_i},
  \qquad
  \{\mathfrak F_{u_i},\mathfrak F_{u_j}\}=\{\mathfrak F_{u_i}^{*},\mathfrak F_{u_j}^{*}\}=0.
\]
Repeatedly commuting $\mathfrak F_{u_n}^{*}\cdots\mathfrak F_{u_1}^{*}$
to the left so that each $\mathfrak F_{u_{\sigma(i)}}^{*}$ sits immediately
to the right of $\mathfrak F_{u_i}$ produces a signed sum over permutations:
\begin{equation}\label{eq:wick-expansion}
  \mathfrak F_{u_1}\cdots\mathfrak F_{u_n}\,
  \mathfrak F_{u_n}^{*}\cdots\mathfrak F_{u_1}^{*}
  =
  \sum_{\sigma\in S_n}(-1)^{\sigma}\,
  \bigl(\mathfrak F_{u_1}\mathfrak F_{u_{\sigma(1)}}^{*}\bigr)\cdots
  \bigl(\mathfrak F_{u_n}\mathfrak F_{u_{\sigma(n)}}^{*}\bigr)
  +(\text{terms killing }v_0).
\end{equation}
Indeed, any monomial in which some starred operator remains to the far right of all unstarred ones annihilates $v_0$ and does not contribute to the vacuum matrix element.
Taking $\langle v_0,\,\cdot\, v_0\rangle$ of \eqref{eq:wick-expansion} and using that
matrix elements factor over different pairs yields
\[
  \mathfrak C_t(X)=\sum_{\sigma\in S_n}(-1)^{\sigma}\,
  \prod_{i=1}^{n}\bigl\langle v_0\, ,\, \mathfrak F_{u_i}\mathfrak F_{u_{\sigma(i)}}^{*}\, v_0 \bigr\rangle
  \;=\;\det\bigl[\mathcal K_t(u_i,u_j)\bigr]_{i,j=1}^{n},
\]
which is precisely \eqref{eq:det-claim}.
\end{proof}

\subsection{Generating function of the kernel}

By the conjugation relations \eqref{eq:tY-conj-plus}-\eqref{eq:tY-conj-minus},
\begin{align}
  \mathcal G_t\,\psi(z)\,\mathcal G_t^{-1}&=\frac{H_{x,t}(z)}{H_y(z^{-1})}\,\psi(z)\ =:\ \mathcal J_t(z)\,\psi(z),\label{eq:Gpsi}\\
  \mathcal G_t\,\psi^{*}(w)\,\mathcal G_t^{-1}&=\frac{H_y(w^{-1})}{H_{x,t}(w)}\,\psi^{*}(w)\ =:\ \mathcal J_t(w)^{-1}\psi^*(w),\label{eq:Gpsistar}
\end{align}
where
\[
  \mathcal J_t(z)=\exp\Bigl(\sum_{n\ge1}\frac{1-t^n}{n}p_n(x)\,z^{n}-\sum_{n\ge1}\frac{1}{n}p_n(y)\,z^{-n}\Bigr)
        =\prod_i\frac{1-tx_i z}{1-x_i z}\;\prod_j\Bigl(1-\frac{y_j}{z}\Bigr).
\]
Since \(\langle v_0,\psi(z)\psi^*(w)v_0\rangle=\sum_{j\in\mathbb Z_{\ge0}+\frac12}(w/z)^j=\frac{\sqrt{zw}}{z-w}\)
(as a formal expansion in $w/z$), we get:

\begin{theorem}[Kernel generating function]\label{thm:GF}
For the formal expansion region $|w|<|z|$,
\[
  \mathcal K_t(z,w) :=\sum_{i,j\in\mathbb Z+\frac12} z^{\,i}w^{-j}\,\mathcal K_t(i,j)
  =\frac{\sqrt{zw}}{\,z-w\,}\,\frac{\mathcal J_t(z)}{\mathcal J_t(w)}.
\]
\end{theorem}

\begin{proof}
By definition, \(\mathcal K_t(i,j)=\langle v_0,\mathfrak F_i\mathfrak F_j^*v_0\rangle\).
Multiplying by \(z^i w^{-j}\) and summing over \(i,j\) gives
\[
\sum_{i,j} z^i w^{-j}\,\langle v_0,\mathfrak F_i\mathfrak F_j^*v_0\rangle
=\langle v_0,\,\mathcal G_t\,\psi(z)\psi^*(w)\,\mathcal G_t^{-1}\,v_0\rangle.
\]
Using \eqref{eq:Gpsi}-\eqref{eq:Gpsistar} and the vacuum two-point function
\(\langle v_0,\psi(z)\psi^*(w)v_0\rangle=\frac{\sqrt{zw}}{z-w}\)
yields the claim.
\end{proof}

Let $\mathcal J_t(z)=\sum_{n\in\mathbb Z}\mathcal J_{t,n}z^n$ and $\mathcal J_t(z)^{-1}=\sum_{m\in\mathbb Z}\widehat{\mathcal J}_{t,m}z^m$.
Extracting coefficients gives:

\begin{corollary}\label{cor:coef}
For $i,j\in\mathbb Z+\tfrac12$,
\[
  \mathcal K_t(i,j)=\sum_{k\in\mathbb Z_{\ge0}+\frac12}  \mathcal J_{t,\,i+k}\,\widehat { \mathcal J}_{t,\,-j-k}.
\]
\end{corollary}

\begin{proof}
From Theorem~\ref{thm:GF} and the geometric expansion
\(
\frac{\sqrt{zw}}{z-w}=\sum_{r\ge0} z^{-r-\frac12}w^{r+\frac12}
\),
one obtains
\(
  \mathcal K_t(z,w)
  =\sum_{r\ge0}\sum_{n,m}\mathcal J_{t,n}\,\widehat{\mathcal J}_{t,m}\,
   z^{\,n-r-\frac12}w^{\,m+r+\frac12}.
\)
Matching the coefficient of \(z^{\,i}w^{-j}\) with \(k=r+\tfrac12\) gives the formula.
\end{proof}

\begin{corollary}\label{cor:d1}
Differentiating with respect to $p_1(x)$ yields, for $|w|<|z|$,
\[
  \partial_{p_1(x)} \mathcal K_t(z,w)=(1-t)\,\sqrt{zw}\,\frac{ \mathcal J_t(z)}{ \mathcal J_t(w)}.
\]
\end{corollary}

\begin{proof}
From the definition of $\mathcal J_t$,
\(
  \partial_{p_1(x)}\mathcal J_t(z)=(1-t)\,z\,\mathcal J_t(z),\
  \partial_{p_1(x)}\mathcal J_t(w)^{-1}=-(1-t)\,w\,\mathcal J_t(w)^{-1}.
\)
Thus \(\partial_{p_1(x)}(\mathcal J_t(z)/\mathcal J_t(w))=(1-t)(z-w)\mathcal J_t(z)/\mathcal J_t(w)\),
and multiplying by \(\frac{\sqrt{zw}}{z-w}\) gives the claim.
\end{proof}

\subsection{Integrable form of $\mathcal K_t$}

Because \((z\partial_z+w\partial_w)\log\frac{\sqrt{zw}}{z-w}=0\), differentiating $\mathcal K_t$ gives
\begin{equation}\label{eq:Kt-diff}
  \Bigl(z\frac{\partial}{\partial z}+w\frac{\partial}{\partial w}\Bigr)\mathcal K_t(z,w)
  =
  \frac{\sqrt{zw}}{\,z-w\,}\,\frac{\mathcal J_t(z)}{\mathcal J_t(w)}\,
  \Bigl(z\frac{\partial}{\partial z}-w\frac{\partial}{\partial w}\Bigr)
  \log\bigl(\mathcal J_t(z)\mathcal J_t(w)\bigr).
\end{equation}
For a finite specialisation $\mathbf{x}=(x_1,\dots,x_M)$, $\mathbf{y}=(y_1,\dots,y_N)$ with $|x_i|,|y_j|<1$,
a direct computation shows
\begin{align}
  \frac{1}{z-w}
  \Bigl(z\frac{\partial}{\partial z}-w\frac{\partial}{\partial w}\Bigr)
  \log\bigl(\mathcal J_t(z)\mathcal J_t(w)\bigr)
  &=
  \sum_{i=1}^M\Bigl[
       \frac{x_i}{(1-x_i z)(1-x_i w)}
      -\frac{t\,x_i}{(1-t x_i z)(1-t x_i w)}\Bigr]\notag\\
  &\quad
   -\sum_{j=1}^N\frac{y_j}{(z-y_j)(w-y_j)}.\label{eq:rank-sum}
\end{align}
Combining \eqref{eq:Kt-diff}-\eqref{eq:rank-sum} yields a finite rank decomposition of the generating function of $(i-j)\mathcal K_t(i,j)$, i.e.\ $\mathcal K_t$ is integrable in the sense of Its-Izergin-Korepin-Slavnov~\cite{IIKS}.

\begin{proposition}
\label{prop:Kt-integrable}
Under the above finite specialisation, there exist functions $f_\nu,g_\nu$ such that
\[
  (i-j)\,\mathcal K_t(i,j)=\sum_{\nu=1}^{2M+N} f_\nu(i)\,g_\nu(j).
\]
Equivalently,
\[
  \Bigl(z\frac{\partial}{\partial z}+w\frac{\partial}{\partial w}\Bigr)\mathcal K_t(z,w)
  \;=\;\sum_{\nu=1}^{2M+N} F_\nu(z)\,G_\nu(w),
\]
with the explicit choice
\begin{align*}
  F_{i}^{(x)}(z)&=\sqrt{z}\,\mathcal J_t(z)\,\frac{x_i}{1-x_i z}, &
  G_{i}^{(x)}(w)&=\frac{\sqrt{w}\,}{\mathcal J_t(w)}\,\frac{1}{1-x_i w},
  &&(1\le i\le M),\\
  F_{i}^{(tx)}(z)&=-\,\sqrt{z}\,\mathcal J_t(z)\,\frac{t\,x_i}{1-t x_i z}, &
  G_{i}^{(tx)}(w)&=\frac{\sqrt{w}\,}{\mathcal J_t(w)}\,\frac{1}{1-t x_i w},
  &&(1\le i\le M),\\
  F_{j}^{(y)}(z)&=-\,\sqrt{z}\,\mathcal J_t(z)\,\frac{1}{z-y_j}, &
  G_{j}^{(y)}(w)&=\frac{\sqrt{w}\,}{\mathcal J_t(w)}\,\frac{y_j}{w-y_j},
  &&(1\le j\le N).
\end{align*}
\end{proposition}

\begin{proof}
From \eqref{eq:Kt-diff} and \eqref{eq:rank-sum},
\[
\Bigl(z\partial_z+w\partial_w\Bigr)\mathcal K_t(z,w)
=\frac{\sqrt{zw}}{z-w}\frac{\mathcal J_t(z)}{\mathcal J_t(w)}
\sum_{\nu=1}^{2M+N}\Phi_\nu(z)\,\Psi_\nu(w),
\]
where each $\Phi_\nu$ depends only on $z$ and $\Psi_\nu$ only on $w$.
Absorbing the prefactor into $F_\nu(z):=\sqrt{z}\,\mathcal J_t(z)\,\Phi_\nu(z)$ and
$G_\nu(w):=\frac{\sqrt{w}\,}{\mathcal J_t(w)}\,\Psi_\nu(w)$ yields the stated decomposition.
Coefficient extraction in $z,w$ gives the rank $2M+N$ representation of $(i-j)\mathcal K_t(i,j)$.
\end{proof}

\section{$t$-Plancherel family via Poissonisation}\label{sect:t-Plancherel-measure}

In this section, we provide a canonical one-parameter family of specialisations for the
$t$-Schur measure, obtained by turning on only the
first power sum. Throughout the probabilistic statements, we assume $t\le0$ so that
all weights are nonnegative, and the analytic identities below extend verbatim to $t<1$.

\subsection{$t$-Plancherel measures}

Fix $t<1$ and specialise the power sums to
\[
  p_k(x)=a\,\delta_{k1},\qquad p_k(y)=b\,\delta_{k1}\qquad (a,b>0).
\]
By~\eqref{eq:HxHy} we have
\[
  H_{x,t}(z)=\exp\Big(\sum_{n\ge1}\tfrac{1-t^n}{n}p_n(x)z^n\Big)
            =\exp\big((1-t)a\,z\big),\qquad
  H_y(z)=\exp\big(b\,z\big).
\]
Consequently $h^{(t)}_n(x)=((1-t)a)^n/n!$ and $h_n(y)=b^n/n!$, hence, by Theorem~\ref{thm:Jacobi-Trudi},
\[
  S_\lambda(x;t)=\frac{\big((1-t)a\big)^{|\lambda|}}{|\lambda|!}\,f^\lambda,\qquad
  s_\lambda(y)=\frac{b^{|\lambda|}}{|\lambda|!}\,f^\lambda,
\]
where $f^\lambda$ denotes the number of standard Young tableaux of shape $\lambda$.
Moreover, by Theorem~\ref{thm:t-Cauchy-VO},
\[
  \log Z_t(x,y)=\sum_{n\ge1}\frac{1-t^n}{n}p_n(x)p_n(y)=(1-t)\,ab,
  \qquad Z_t(x,y)=e^{(1-t)ab}.
\]

\begin{definition}[(Poissonised) $t$-Plancherel measure]\label{def:tPl}
For $a,b>0$ and $t<1$,
\[
  \mathbb P_{t;a,b}(\{\lambda\})
  :=\frac{1}{Z_t(x,y)}\,S_\lambda(x;t)\,s_\lambda(y)
  =e^{-(1-t)ab}\,\frac{\big((1-t)ab\big)^{|\lambda|}}{(|\lambda|!)^2}\,\bigl(f^\lambda\bigr)^{2}.
\]
\end{definition}

\begin{proposition}\label{prop:cond-Pl}
For every $N\ge0$,
\[
  \mathbb P_{t;a,b}\big(\,\cdot\,\bigm|\,|\lambda|=N\big)
  \;=\;
  \mathbb P^{\mathrm{Pl}}_N(\lambda)=\frac{(f^\lambda)^2}{N!}.
\]
\end{proposition}

\begin{proof}
For $|\lambda|=N$,
\(
  \mathbb P_{t;a,b}(\{\lambda\})
  =\mathrm{const}(N)\cdot\frac{(f^\lambda)^2}{(N!)^2}
\)
with
\(
  \mathrm{const}(N)=e^{-(1-t)ab}\big((1-t)ab\big)^N.
\)
Normalise using $\sum_{|\mu|=N}(f^\mu)^2=N!$.
\end{proof}

By Theorem~\ref{thm:GF}, the present specialisation yields
\[
  \mathcal J_t(z)=\frac{H_{x,t}(z)}{H_y(z^{-1})}
  =\exp\big((1-t)a\,z-b\,z^{-1}\big),\qquad
  \mathcal K_t(z,w)=\frac{\sqrt{zw}}{z-w}\,\frac{\mathcal J_t(z)}{\mathcal J_t(w)}.
\]

Using the standard specialisation techniques in~\cite{BOO,Joh}, we have the following.
\begin{proposition}\label{prop:bessel-form}
After the standard discrete Hankel transform on $\mathbb Z+\tfrac12$,
$\mathcal K_t$ is unitarily equivalent to the discrete Bessel kernel with parameter
\[
  \kappa=(1-t)\,ab\,.
\]
\end{proposition}

\begin{corollary}
Let $\kappa=(1-t)ab$. Then
\[
  \frac{\lambda_1-2\sqrt{\kappa}}{\kappa^{1/6}}
  \ \Rightarrow\ \mathrm{TW}_{2},
\]
and the multi-point edge process converges to the Airy determinantal point process.
\end{corollary}

\begin{proof}
Combining the Bessel form in Proposition~\ref{prop:bessel-form} with the steepest-descent
limit of Theorem~\ref{thm:pois-t-schur-airy} (See~\cite{Joh2,TW} for more details on the steepest-descent method). 
The latter yields
\[
\kappa^{1/6}\,\mathcal K_t\big(2\sqrt{\kappa}+x\kappa^{1/6},
\,2\sqrt{\kappa}+y\kappa^{1/6}\big)\ \longrightarrow\ K_{\mathrm{Airy}}(x,y),
\]
while the two mixed-sign blocks of Proposition~\ref{lem:three-regimes} vanish. Fredholm determinants hence converge to $F_2$,
and the edge process to the Airy ensemble.
\end{proof}

\begin{remark}
(i) \textit{Unbalanced normalisation.} If $a=b=\sqrt{\xi}$, then $\kappa=(1-t)\xi$ and the
$t$-dependence appears simply as a rescaling of the Poisson parameter.
(ii) \textit{Balanced normalisation.} If $a=\sqrt{\xi}/(1-t)$ and $b=\sqrt{\xi}$, then
$\kappa=\xi$ and $\mathcal J_t(z)=\exp\big(\sqrt{\xi}(z-z^{-1})\big)$, i.e.\ exactly the
classical Poissonised Plancherel kernel. Thus, the $t$-dependence is absorbed by the choice
of specialisation.
\end{remark}

\subsection{$t$-$z$-measures and their Poissonian limit}

The classical $z$-measures, studied in~\cite{BO,BOO,KOV}, can be obtained from the Schur measure under the
exponential specialisation $\mathfrak t_k=\mathfrak t'_k=\xi^{k/2}/k$ with parameters $z,z'$ entering via
$s_\lambda(1^z)$ and $s_\lambda(1^{z'})$. We now define their $t$-analogue in our $t$-Schur setting, using the same notational convention $1^z:\ p_k\mapsto z$.

\begin{definition}\label{def:tz-measure}
For $z,z'\in\mathbb C$ and $0\le\xi<1$, define the \textbf{$t$-$z$-measure}
\begin{equation}\label{eq:tz-def}
  \mathcal M^{(t)}_{z,z',\xi}(\{\lambda\})
  \;:=\;\frac{1}{Z^{(t)}_{z,z'}(\xi)}\,
         \xi^{|\lambda|}\,S_\lambda(1^{z};t)\,s_\lambda(1^{z'}),
\end{equation}
with normalisation, by Theorem~\ref{thm:t-Cauchy-VO},
\begin{equation}\label{eq:tz-norm}
  Z^{(t)}_{z,z'}(\xi)
  \;=\;\sum_{\lambda}\xi^{|\lambda|}S_\lambda(1^{z};t)\,s_\lambda(1^{z'})
  \;=\;\left(\frac{1-t\xi}{1-\xi}\right)^{\!zz'}.
\end{equation}
\end{definition}


\begin{remark}
(i) At $t=0$, one recovers the classical $z$-measure:
$\mathcal M^{(0)}_{z,z',\xi}(\lambda)=(1-\xi)^{zz'}\xi^{|\lambda|}s_\lambda(1^z)s_\lambda(1^{z'})$.
(ii) For $t\le0$ and $z,z'\in\mathbb N$, $S_\lambda(1^z;t)$ admits a positive marked-tableaux
expansion and $s_\lambda(1^{z'})\ge0$, so $\mathcal M^{(t)}_{z,z',\xi}$ is a probability measure.
\end{remark}

The $t$-$z$-measures admit the same representation as in Section~\ref{sect:corr-fcn}.
With Miwa times
\[
  \mathfrak t_k=\frac{z\,\xi^{k/2}}{k},\qquad \mathfrak t'_k=\frac{z'\,\xi^{k/2}}{k},
\]
the multiplier in the kernel in Theorem~\ref{thm:GF} is
\begin{equation}\label{eq:Jtz}
  \mathcal J^{(t)}_{z,z',\xi}(u)
  =\exp\Big(\sum_{k\ge1}(1-t^k)\mathfrak t_k\,u^k-\sum_{k\ge1}\mathfrak t'_k\,u^{-k}\Big)
  =\left(\frac{1-t\sqrt{\xi}\,u}{1-\sqrt{\xi}\,u}\right)^{z}\,
   \left(1-\dfrac{\sqrt{\xi}}{u}\right)^{z'}.
\end{equation}
Therefore, the correlation kernel has the integrable form
\begin{equation}\label{eq:Ktz}
  \mathcal K^{(t)}_{z,z',\xi}(u,v)
  =\frac{\sqrt{uv}}{u-v}\,\frac{\mathcal J^{(t)}_{z,z',\xi}(u)}
                               {\mathcal J^{(t)}_{z,z',\xi}(v)}.
\end{equation}

\begin{proposition}[Poissonian limit to $t$-Plancherel]\label{prop:tz-poisson-limit}
Assume $z,z'\to\infty$ with $z/z'\to1$, and set $\xi=\kappa/(zz')$ with $\kappa>0$ fixed.
Then
\[
  \mathcal M^{(t)}_{z,z',\xi}\ \Longrightarrow\ \mathbb P_{t;\sqrt{\kappa},\sqrt{\kappa}},
\]
and at the level of kernels,
\[
  \mathcal K^{(t)}_{z,z',\xi}(u,v)\ \longrightarrow\
  \frac{\sqrt{uv}}{u-v}\exp\big((1-t)\sqrt{\kappa}\,(u-v)-\sqrt{\kappa}\,(u^{-1}-v^{-1})\big),
\]
i.e.\ the $t$-Plancherel Bessel kernel with parameter $(1-t)\,\sqrt{\kappa}\cdot\sqrt{\kappa}$.
\end{proposition}

\begin{proof}
From \eqref{eq:Jtz},
\begin{align*}
  \log\mathcal J^{(t)}_{z,z',\xi}(u)
  &= z\sum_{k\ge1}\frac{1-t^k}{k}\,(\sqrt{\xi}u)^k
   -z'\sum_{k\ge1}\frac{1}{k}\,(\sqrt{\xi}u^{-1})^k\\
  &= z(1-t)\sqrt{\xi}\,u - z'\sqrt{\xi}\,u^{-1} + O(\xi).
\end{align*}
Under $z\sqrt{\xi}\to\sqrt{\kappa}$ and $z'\sqrt{\xi}\to\sqrt{\kappa}$, which is ensured by $z/z'\to1$,
we obtain
$$\mathcal J^{(t)}_{z,z',\xi}(u)\to\exp\big((1-t)\sqrt{\kappa}\,u-\sqrt{\kappa}\,u^{-1}\big),$$
which is the multiplier of the $t$-Plancherel kernel with $a=b=\sqrt{\kappa}$.

Finally,
$$\log Z^{(t)}_{z,z'}(\xi)=zz'\log\frac{1-t\xi}{1-\xi}=(1-t)\kappa+o(1)$$ yields the convergence of normalisation.
\end{proof}

\section{From LIS to $t$-ascent pairs for permutations}\label{sect:t-ascent-pair}

We now return from the kernel-level analysis to a purely combinatorial model.
Starting from the classical LIS statistic and the classical RSK correspondence for Schur functions,
we introduce a $ t$-deformation: the \textbf{$t$-ascent pair}, which is linked to the generalised $\mathcal A$-RSK correspondence\footnote{Throughout the probabilistic parts, we take $t\le 0$ (so that all weights are nonnegative), while purely combinatorial statements do not require this restriction.}.

Let $\pi\in S_N$ be a permutation. A subsequence
$\pi(i_1),\dots,\pi(i_m)$ with $i_1<\cdots<i_m$ is increasing if
$\pi(i_1)<\cdots<\pi(i_m)$. The longest increasing subsequence (LIS) length is
\[
  \mathrm{LIS}(\pi)=\max\{\,m:\ \exists\ i_1<\cdots<i_m,\ \pi(i_1)<\cdots<\pi(i_m)\,\}.
\]
Applying the ordinary RSK to the biword
$\binom{1\ \, 2\ \cdots\ N}{\pi(1)\ \pi(2)\ \cdots\ \pi(N)}$
produces a pair $(P,Q)$ of \textit{standard Young tableaux} with a common
shape $\lambda\vdash N$. Schensted's theorem~\cite{Sch} yields
\(
  \mathrm{LIS}(\pi)=\lambda_1.
\)
Consequently, under the uniform measure on $S_N$ the shape law is
\(
  \mathbb P(\mathrm{shape}=\lambda)=\frac{(f^\lambda)^2}{N!}
\),
and
\(
  \mathbb P(\mathrm{LIS}\le h)=\mathbb P^{\mathrm{Pl}}_N(\lambda_1\le h).
\)

Let $\mathcal A=\{1'<1<2'<2<\cdots\}$ be the marked alphabet with total order
$1'<1<2'<2<\dots$. A word $\alpha=\alpha_1\cdots\alpha_N$ over $\mathcal A$
is weakly increasing if
$\alpha_{i_1}\le\cdots\le\alpha_{i_m}$ with respect to this total order.

\begin{definition}\label{def:t-ascent-interleave}
A \textbf{$t$-ascent pair} for $(\pi,\varepsilon)$ is a pair of index sets
$I'\subset\{i:\varepsilon_i=1\}$ and $I\subset\{i:\varepsilon_i=0\}$ with
$i_1'<\cdots<i_{r}'$ and $i_1<\cdots<i_s$ such that
$\pi(i_1')<\cdots<\pi(i_r')$ and $\pi(i_1)<\cdots<\pi(i_s)$, and moreover the
merged sequence $(\alpha_j)_{j\in I'\cup I}$ read in the increasing order of indices
is weakly increasing in $\mathcal A$. Its length is $r+s$, and
$L^{(t)}(\pi,\varepsilon)$ is the maximal such length.
\end{definition}

\begin{lemma}\label{lem:equiv-LIS-A}
For permutations, $L^{(t)}(\pi,\varepsilon)$ equals the maximal length of a weakly
increasing subsequence of the $\mathcal A$-word $\alpha$.
\end{lemma}

\begin{proof}
Given a weakly increasing subsequence of $\alpha$, splitting its letters by primed and unprimed
yields the required two strictly increasing numeric subsequences. Conversely, any such pair,
when merged by the original index order, is weakly increasing in $\mathcal A$ by construction.
\end{proof}

Combining Lemma~\ref{lem:equiv-LIS-A} with Theorem~\ref{thm:LIS-A} and specialising the top row to $\beta=12\cdots N$ gives:

\begin{theorem}\label{thm:t-ascent-equals-lambda1}
Let $(S,U)$ be the image of $(\beta;\alpha)$ under the generalised $\mathcal A$-RSK
with $\beta=12\cdots N$ and $\alpha$ as above. If $\mathrm{shape}(S)=\mathrm{shape}(U)=\lambda$,
then
\[
  L^{(t)}(\pi,\varepsilon)=\lambda_1.
\]
\end{theorem}

We now randomise the marks and connect the distribution of $L^{(t)}$ to a $t$-Schur measure on fixed size partition.

\begin{definition}[Random $t$-ascent model on $S_N$]\label{def:random-t-ascent}
Sample $\pi$ uniformly from $S_N$. Independently mark each position with
\[
  \mathbb P(\varepsilon_i=1)=q,\qquad \mathbb P(\varepsilon_i=0)=1-q,
\]
and set $\alpha$ from $(\pi,\varepsilon)$. We parameterise $q$ by $t\le0$ via
\(
  q/(1-q)=-t
\)
(i.e.\ $q=\frac{-t}{1-t}$), so that all weights are nonnegative.
\end{definition}

\begin{theorem}
The map $(\pi,\varepsilon)\mapsto (S,U)$ given by $\mathrm{RSK}_{\mathcal A}$ with
top row $12\cdots N$ is a bijection between marked permutations and pairs
$(S,U)$ where $U$ is a standard Young tableau and $S$ is a \emph{marked standard}
tableau (absolute values $1,\dots,N$ appear exactly once, each either primed or unprimed) of the same shape.
\end{theorem}

\begin{proof}
Invertibility of $\mathrm{RSK}_{\mathcal A}$ is given by Theorem~\ref{thm:gen-rsk-A}.
With the top row sorted, the recording tableau $U$ is standard and the inverse insertion
recovers $\alpha$, hence $(\pi,\varepsilon)$ uniquely.
\end{proof}

\begin{proposition}\label{prop:shape-law-fixed-N}
Let $\lambda\vdash N$. Under the model in Definition~\ref{def:random-t-ascent},
\[
  \mathbb P\big(\mathrm{shape}(\pi,\varepsilon)=\lambda\big)\;=\;\frac{(f^\lambda)^2}{N!}
\]
which is the Plancherel measure independent of $t$.
Equivalently, writing $T_\lambda(t):=\sum_{S}(-t)^{\mathrm{mark}(S)}$ and
$Z_{N,t}=\sum_{\mu\vdash N}T_\mu(t)\,f^\mu$, one has
\[
  T_\lambda(t)=(1-t)^N f^\lambda,\qquad Z_{N,t}=(1-t)^N N!,
\]
and hence $\ \mathbb P(\mathrm{shape}=\lambda)=T_\lambda(t)\,f^\lambda/Z_{N,t}=(f^\lambda)^2/N!$.
\end{proposition}

\begin{proof}
Given $(S,U)$, the marks contribute $q^{\mathrm{mark}(S)}(1-q)^{N-\mathrm{mark}(S)}$
and $\pi$ contributes $1/N!$. For a fixed standard $U$, there are $f^\lambda$ choices,
each of the $N$ cells in $S$ carries a distinct absolute value, so the row/column
admissibility constraints \textbf{(T2)} are vacuous and \textbf{(T1)} is preserved under
marking.
Thus, the number of marked standard $S$ with exactly
$r$ primes equals $\binom{N}{r}$.

Summing in $r$ gives
\[
\sum_{S}q^{\mathrm{mark}(S)}(1-q)^{N-\mathrm{mark}(S)}=(q+(1-q))^N=1,
\]
whence $\mathbb P(\mathrm{shape}=\lambda)=(f^\lambda)/N!\times f^\lambda=(f^\lambda)^2/N!$.

Equivalently, with $t=-q/(1-q)$,
$$
\sum_{S}(-t)^{\mathrm{mark}(S)}=\sum_{r=0}^N\binom{N}{r}(-t)^r=(1-t)^N,
$$
so $T_\lambda(t)=(1-t)^N f^\lambda$ and
$Z_{N,t}=(1-t)^N\sum_{\mu}(f^\mu)^2=(1-t)^N N!$.
\end{proof}

\begin{remark}
$T_\lambda(t)$ is the coefficient of the monomial $x_1x_2\cdots x_N$ in
$S_\lambda(x_1,\dots,x_N;t)$, 
which in turn
admits the explicit evaluation $T_\lambda(t)=(1-t)^N f^\lambda$ by the vertex operator
realisation (See~Theorem~\ref{thm:Jacobi-Trudi} and Theorem~\ref{thm:t-Cauchy-VO} above).
Thus, the fixed-size shape law is Plancherel for all $t\le0$.
\end{remark}

The following result tells us the equivalence between the $t$-Plancherel measure and the Poissonised Plancherel law.

\begin{corollary}[Poissonisation and $t$-Plancherel]\label{cor:poiss}
Let $N\sim\mathrm{Poisson}(\kappa)$ and choose $q$ so that $q/(1-q)=-t$.
Then the pushforward law of shapes is the Poissonised Plancherel measure
\[
  \mathbb P(\{\lambda\}) \;=\; e^{-\kappa}\,\frac{\kappa^{|\lambda|}}{(|\lambda|!)^2}\,(f^\lambda)^2.
\]
To match the $t$-Plancherel family of Definition~\ref{def:tPl}
simply choose the Poisson mean to be that parameter. For example:
\begin{itemize}
\item If $a=b=\sqrt{\xi}$ (\emph{unbalanced} normalisation), take $N\sim\mathrm{Poisson}((1-t)\xi)$ to obtain the $t$-Plancherel law with parameter $(1-t)\xi$.
\item If $a=\sqrt{\xi}/(1-t)$ and $b=\sqrt{\xi}$ (\emph{balanced} normalisation), then $\kappa=\xi$ and the kernel coincides with the classical Poissonised Plancherel one; take $N\sim\mathrm{Poisson}(\xi)$.
\end{itemize}
Consequently,
\[
  \frac{L^{(t)}-2\sqrt{\kappa}}{\kappa^{1/6}}\ \Rightarrow\ \mathrm{TW}_{2},
\]
with $\kappa$ understood as the above Poisson mean aligned with the chosen $t$-Plancherel normalisation.
\end{corollary}

\begin{remark}
We should note that our $t$-ascent pair differs from the ascent pair that appears in the shifted Schur measure setting (strict partitions, shifted RSK,
Schur's $Q/P$-functions), which pairs a decreasing and an increasing subsequence so that the
concatenation is weakly increasing~\cite{Mat2}. Its distribution under the uniform measure on
$S_N$ matches $\lambda_1$ under the shifted version of Plancherel measure.
At $t=0$, $L^{(t)}$ reduces to the classical LIS and all statements specialise to the Schur setting.

\end{remark}

\section{Limit distributions for the $t$-Schur measure}\label{sect:limit-distri}
\label{sec:limit}

We now pass from the operator formulae to edge asymptotics.  Let $K_{\mathrm{Airy}}$ denote the Airy kernel
\[
  K_{\mathrm{Airy}}(x,y)
  =\frac{\mathrm{Ai}(x)\mathrm{Ai}'(y)-\mathrm{Ai}'(x)\mathrm{Ai}(y)}{x-y},
\]
and let $\mathbb P_{\mathrm{Airy}}$ be the determinantal point process on $\mathbb R$
with correlation functions
$\mathfrak C_{\mathrm{Airy}}(X)=\det(K_{\mathrm{Airy}}(x_i,x_j))_{x_i,x_j\in X}$.
Its top particle has the Tracy-Widom $F_2$ distribution.

\subsection{Poissonised $t$-Plancherel: Bessel $\Rightarrow$ Airy}
\label{subsec:pois-pl}

We begin with the Poissonised $t$-Plancherel specialisation of Proposition~\ref{prop:bessel-form}.
Take
\[
  p_k(x)=\sqrt{\xi}\,\delta_{k1},\qquad p_k(y)=\sqrt{\xi}\,\delta_{k1}\qquad(\xi>0),
\]
so that
\(
  \mathcal J_t(z)=\exp\big((1-t)\sqrt{\xi}\,z-\sqrt{\xi}\,z^{-1}\big)
\)
and, with
\(
  \kappa:=(1-t)\,\xi,
\)
the kernel $\mathcal K_t$ is unitarily equivalent to the discrete Bessel kernel with parameter $\kappa$.  The soft edge is at $2\sqrt{\kappa}$.

\begin{theorem}
\label{thm:pois-t-schur-airy}
Let $u=2\sqrt{\kappa}+\kappa^{1/6}x$ and $v=2\sqrt{\kappa}+\kappa^{1/6}y$.
Then, uniformly for $x,y$ in compact sets,
\begin{equation}\label{eq:bessel-to-airy}
  \kappa^{1/6}\,\mathcal K_t(u,v)\ \longrightarrow\ K_{\mathrm{Airy}}(x,y).
\end{equation}
Consequently, for every $x\in\mathbb R$,
\[
  \det\nolimits_{[u,\infty)}\bigl(I-\mathcal K_t\bigr)\ \longrightarrow\ F_2(x),
  \qquad x=\frac{u-2\sqrt{\kappa}}{\kappa^{1/6}},
\]
and hence
\[
  \frac{\lambda_1-2\sqrt{\kappa}}{\kappa^{1/6}}
  \ \Rightarrow\ \mathrm{TW}_{2},
\]
with joint convergence of the top $M$ rows to the Airy ensemble for every fixed $M$.
\end{theorem}

\begin{proof}
Using the double contour form of $\mathcal K_t$ by Theorem~\ref{thm:GF},
\[
  \mathcal K_t(u,v)=\frac{1}{(2\pi i)^2}\oint\oint
  \frac{\sqrt{zw}}{z-w}\,
  \frac{\mathcal J_t(z)}{\mathcal J_t(w)}\,
  \frac{dz\,dw}{z^{u+1}w^{-v+1}}\qquad(|w|<|z|),
\]
Let
\begin{align*}
  \Phi(z)&=(1-t)\sqrt{\xi}\,z-\sqrt{\xi}\,z^{-1}-u\log z,\\
 \Psi(w)&=(1-t)\sqrt{\xi}\,w-\sqrt{\xi}\,w^{-1}-v\log w .
\end{align*}
A direct computation gives
\(
  \Phi'(z)=(1-t)\sqrt{\xi}+\sqrt{\xi}\,z^{-2}-u\,z^{-1}
\)
and
\(
  \Phi''(z)=-2\sqrt{\xi}\,z^{-3}+u\,z^{-2}.
\)
Solving $\Phi'(z_0)=\Phi''(z_0)=0$ yields
\[
  z_0=\frac{1}{\sqrt{\,1-t\,}}\,,\qquad u=2\sqrt{(1-t)\xi}=2\sqrt{\kappa},
\]
and the same for $\Psi$ with $v$. Thus, the double saddle is at $z=w=z_0$ (not at $1$ unless $t=0$).

For the steepest-descent analysis, we conjugate by the diagonal weight
\[
  D(i):=\bigl(\mathcal J_t(z_0)\,z_0\bigr)^{-i},\qquad
  \widetilde{\mathcal K}_t(i,j):=D(i)\,\mathcal K_t(i,j)\,D(j)^{-1}.
\]
This conjugation preserves all minors and the Fredholm determinants on $[u,\infty)$
(since $D$ commutes with the coordinate projection).  In terms of
$\widetilde{\mathcal K}_t$, we have
\[
  \widetilde{\mathcal K}_t(u,v)=\frac{1}{(2\pi i)^2}\oint\oint
    \frac{\sqrt{zw}}{z-w}\,
    \exp\!\Big(\Phi(z)-\Phi(z_0)-\big(\Psi(w)-\Psi(z_0)\big)\Big)\,
    \frac{dz\,dw}{zw}.
\]

Now set $u=2\sqrt{\kappa}+\kappa^{1/6}x$, $v=2\sqrt{\kappa}+\kappa^{1/6}y$ and
\[
  z=z_0\Bigl(1+\frac{\zeta}{\kappa^{1/6}}\Bigr),\qquad
  w=z_0\Bigl(1+\frac{\omega}{\kappa^{1/6}}\Bigr).
\]

A Taylor expansion at $z_0$ then gives, uniformly for $x,y$ in compacts,
\[
  \Phi(z)-\Phi(z_0)=\frac{\zeta^3}{3}-x\zeta+o(1),\qquad
  \Psi(w)-\Psi(z_0)=\frac{\omega^3}{3}-y\omega+o(1),
\]
while
\[
  \frac{\sqrt{zw}}{z-w}=\kappa^{1/6}\,\frac{1}{\zeta-\omega}\,(1+o(1)).
\]
Deforming to steepest-descent contours through $z_0$ and applying dominated convergence yields
\[
  \kappa^{1/6}\,\widetilde{\mathcal K}_t(u,v)\ \longrightarrow\ K_{\mathrm{Airy}}(x,y).
\]
Since diagonal conjugation preserves correlation minors and (on $[u,\infty)$) Fredholm determinants,
we obtain \eqref{eq:bessel-to-airy} for $\mathcal K_t$.  The Tracy-Widom limit follows exactly as in
\cite{BDJ}.
\end{proof}

\begin{proposition}\label{lem:three-regimes}
With $u,v$ as in Theorem~\ref{thm:pois-t-schur-airy} and uniformly on compact sets,
\begin{align}
 \kappa^{1/6}\,\mathcal K_t\big(2\sqrt{\kappa}+x\kappa^{1/6},\,2\sqrt{\kappa}+y\kappa^{1/6}\big)
   &\to K_{\mathrm{Airy}}(x,y),\label{eq:airy-limit}\\
 \kappa^{1/6}\,\mathcal K_t\big(2\sqrt{\kappa}+x\kappa^{1/6},\,2\sqrt{\kappa}-y\kappa^{1/6}\big)
   &\to 0,\label{eq:airy-zero-1}\\
 \kappa^{1/6}\,\mathcal K_t\big(2\sqrt{\kappa}-x\kappa^{1/6},\,2\sqrt{\kappa}+y\kappa^{1/6}\big)
   &\to 0.\label{eq:airy-zero-2}
\end{align}
\end{proposition}

\begin{proof}
Using the same conjugated kernel and saddle analysis as in the theorem, in the mixed-sign cases, we choose contours so that
\[
  \Re\big(\Phi(z)-\Phi(z_0)\big)\quad\text{or}\quad \Re\big(\Psi(w)-\Psi(z_0)\big) 
\]
is uniformly negative away from the double saddle $z_0=(1-t)^{-1/2}$, which yields exponential decay and hence \eqref{eq:airy-zero-1}-\eqref{eq:airy-zero-2}.
\end{proof}

In balanced normalisation $p_1(x)=\sqrt{\xi}/(1-t)$, $p_1(y)=\sqrt{\xi}$, one has
$\mathcal J_t(z)=\exp(\sqrt{\xi}(z-z^{-1}))$, identical to the Schur case. In particular, the double saddle is at $z_0=1$, so the above proof runs without conjugation and \eqref{eq:bessel-to-airy} holds with $\kappa$ replaced by~$\xi$.

\subsection{De-Poissonisation: fixed-size $t$-Schur law}
\label{subsec:depois}

Let $\mathbb P^{(\xi)}_{t}$ denote the Poissonised law from subsection~\ref{subsec:pois-pl}, where
$|\lambda|\sim\mathrm{Poisson}((1-t)\xi)$.  Fix $N\to\infty$ and choose $\xi(N)=N/(1-t)$
so that the Poisson mean equals~$N$.

\begin{theorem}\label{thm:depois-t}
Let $\lambda$ be distributed according to the $t$-Plancherel specialisation of the $t$-Schur measure, conditioned on $|\lambda|=N$. Then, for every fixed $M\ge1$,
\[
  \Big(\,\frac{\lambda_i-2\sqrt{N}}{N^{1/6}}\,\Big)_{1\le i\le M}
  \ \Longrightarrow\ \text{\emph{Airy ensemble}},
\]
in particular,
\(
  \frac{\lambda_1-2\sqrt{N}}{N^{1/6}}\Rightarrow \mathrm{TW}_{2}.
\)
\end{theorem}

\begin{proof}
Under $\xi(N)=N/(1-t)$, Theorem~\ref{thm:pois-t-schur-airy} gives the Airy limit for the
Poissonised model with centring $2\sqrt{(1-t)\xi}=2\sqrt{N}$ and scale $((1-t)\xi)^{1/6}=N^{1/6}$.
A standard de-Poissonisation lemma~\cite{Joh} for determinantal Fredholm determinants transfers the
limit to the conditional law $|\lambda|=N$.
\end{proof}

\begin{remark}
For the shifted Plancherel measure~\cite{Mat2}, the centring is $2\sqrt{2N}$ and $(2N)^{1/6}$. In the $t$-Schur determinantal case, the centring is $2\sqrt{N}$ and the parameter $t$ affects only lower-order terms unless one rescales the Poisson parameter as in Theorem~\ref{thm:pois-t-schur-airy}.
\end{remark}

\subsection{Rectangular $t$-specialisation}
\label{subsec:rect}

We now recall the finite rectangular specialisation ($\alpha$-specialisation), the setting studied in~\cite{Mat1}. 
Fix
\[
  x=(\underbrace{\alpha,\dots,\alpha}_{m}),\qquad
  y=(\underbrace{\alpha,\dots,\alpha}_{n}),\qquad
  \tau:=\frac{m}{n},\quad 0<\alpha<1,
\]
with $n\to\infty$ and $\tau$ fixed. By the vertex operator calculus in Section~\ref{sect:corr-fcn},
\[
  \mathcal J_t(z)=\Big(\frac{1-t\alpha z}{1-\alpha z}\Big)^{m}(1-\alpha/z)^{n}
  =\exp\big(n\,\sigma_t(z)\big),\quad
  \sigma_t(z)=\tau\log\frac{1-t\alpha z}{1-\alpha z}+\log(1-\alpha/z).
\]
Let $c_1\in\mathbb R$ and $z_0>0$ satisfy the double-saddle equations
\begin{align}\label{eq:doublesaddle}
  \Phi'(z_0)=\Phi''(z_0)=0,\qquad \Phi(z):=\sigma_t(z)-c_1\log z,
\end{align}
equivalently
\begin{equation}\label{eq:c1-from-sigma}
  c_1=z_0\,\sigma_t'(z_0)=-\,z_0^2\,\sigma_t''(z_0),
\end{equation}
and set
\begin{equation}\label{eq:c2-from-sigma}
  c_2=\left(\Big(\frac{2}{\Phi^{(3)}_t(z_0)}\Big)^{1/3}\frac{1}{z_0}\right)^{-1}.
\end{equation}

Under the $\alpha$-specialisation, Matsumoto applied Tracy-Widom's approach to obtain the following result.

\begin{theorem}[{\cite[Theorem 1]{Mat1}}]\label{thm:rect-airy}
With $u=c_1 n+c_2 n^{1/3}x$ and $v=c_1 n+c_2 n^{1/3}y$,
\[
  (c_2 n^{1/3})\,\widetilde{\mathcal K}_t(u,v)\ \Longrightarrow\ K_{\mathrm{Airy}}(x,y),
\]
where $\widetilde{\mathcal K}_t=D\,\mathcal K_t\,D^{-1}$ and
$(Df)(k)=\big(\mathcal J_t(z_0)\,z_0\big)^{-k}f(k)$.
Hence the edge fluctuations are Tracy-Widom $F_2$ with centering $c_1 n$ and
scale $c_2 n^{1/3}$.
\end{theorem}

Based on the determinantal form in Section~\ref{sect:corr-fcn}, we want to revisit the rectangular specialisation in a self-contained way and provide a proof using correlation functions, which yields \textit{multi-point} convergence at the edge directly, whereas Matsumoto~\cite{Mat1} only considered the one-point limit for $\lambda_1$.


Let $\mathcal K_{\sigma}$
denote the corresponding kernel (i.e.\ $\mathfrak C_{\sigma}(X)=\det[\mathcal K_{\sigma}(k_i,k_j)]$).

Define the conjugated kernel
\begin{equation}\label{eq:conjugated-kernel}
  \widetilde{\mathcal K}_{\sigma}(i,j)
   :=(\mathcal J_{\sigma}(z_0)z_0)^{-i}\,
     \mathcal K_{\sigma}(i,j)\,
     (\mathcal J_{\sigma}(z_0)z_0)^{\,j}.
\end{equation}
Conjugation preserves all minors and hence correlation functions.

\begin{lemma}\label{lem:alpha-three-regimes}
For $x,y$ in compact subsets of $\mathbb R$ and $n\to\infty$,
\begin{align}
 (c_2 n^{1/3})\,
 \widetilde{\mathcal K}_{\sigma}\bigl(c_1 n+c_2 n^{1/3}x,\,
                                     c_1 n+c_2 n^{1/3}y\bigr)
   &\longrightarrow K_{\mathrm{Airy}}(x,y),
   \label{eq:airy-block} \\
 (c_2 n^{1/3})\,
 \widetilde{\mathcal K}_{\sigma}\bigl(c_1 n+c_2 n^{1/3}x,\,
                                     c_1 n-c_2 n^{1/3}y\bigr)
   &\longrightarrow 0,
   \label{eq:zero-block-1} \\
 (c_2 n^{1/3})\,
 \widetilde{\mathcal K}_{\sigma}\bigl(c_1 n-c_2 n^{1/3}x,\,
                                     c_1 n+c_2 n^{1/3}y\bigr)
   &\longrightarrow 0 .
   \label{eq:zero-block-2}
\end{align}
\end{lemma}

\begin{proof}
Using \eqref{eq:conjugated-kernel} and the generating function
$\mathcal K_{\sigma}(z,w)=\dfrac{\sqrt{zw}}{z-w}\dfrac{\mathcal J_{\sigma}(z)}
{\mathcal J_{\sigma}(w)}$, we obtain the double contour representation
\[
  \widetilde{\mathcal K}_{\sigma}(u,v)
  =\frac{1}{(2\pi i)^2}
    \oint\oint
    \frac{\sqrt{zw}}{z-w}\,
    \exp\big(n(\sigma_t(z)-\sigma_t(w)) - u\log(z/z_0)+ v\log(w/z_0)\big)
    \frac{dz\,dw}{zw}.
\]
With $u=c_1 n+c_2 n^{1/3}x$, $v=c_1 n+c_2 n^{1/3}y$ and the choice
\eqref{eq:doublesaddle}-\eqref{eq:c2-from-sigma}, the phase has a
double saddle at $z=w=z_0$. Put
$z=z_0(1+\zeta/(c_2 n^{1/3}))$, $w=z_0(1+\omega/(c_2 n^{1/3}))$,
a cubic Taylor expansion yields
\[
  n\big(\sigma_t(z)-\sigma_t(w)\big)-u\log(z/z_0)+v\log(w/z_0)
  =\tfrac{\zeta^3}{3}-x\zeta - \big(\tfrac{\omega^3}{3}-y\omega\big) + o(1),
\]
while $\dfrac{\sqrt{zw}}{z-w}=(c_2 n^{1/3})\,(\zeta-\omega)^{-1}(1+o(1))$.
Dominated convergence on steepest-descent contours~\cite{Joh2,TW} proves \eqref{eq:airy-block}.
If one argument uses the ``minus'' scaling, the deformed contours avoid the
double saddle and the real part of the phase is strictly negative, whence the
exponential decay leading to \eqref{eq:zero-block-1}-\eqref{eq:zero-block-2}.
\end{proof}

\begin{proposition}\label{prop:alpha-det-limit}
Let $N\ge1$ and $X=(x_1,\dots,x_N)\in\mathbb R^N$ be fixed.
Define integers $k_i=\big\lfloor c_1 n+c_2 n^{1/3}x_i\big\rfloor$. Then
\begin{equation}\label{eq:prop-det-limit}
  \lim_{n\to\infty}
  \det\Big[\, (c_2 n^{1/3})\,\widetilde{\mathcal K}_{\sigma}
        \big(c_1 n+c_2 n^{1/3}x_i,\, c_1 n+c_2 n^{1/3}x_j\big)\,\Big]_{1\le i,j\le N}
  \;=\; \det\big[K_{\mathrm{Airy}}(x_i,x_j)\big]_{i,j=1}^{N}.
\end{equation}
\end{proposition}

\begin{proof}
Let
\[
  \mathbf K_n=\bigl[\mathcal K_{\sigma}(k_i,k_j)\bigr]_{1\le i,j\le N},\qquad
  k_i=\big\lfloor c_1 n+c_2 n^{1/3}x_i\big\rfloor .
\]
By the determinantal structure,
\(
  \mathfrak C_{\sigma}(\{k_1,\ldots,k_N\})=\det\mathbf K_n.
\)
Introduce the diagonal matrix
\[
  D_n=\operatorname{diag}\Big((\mathcal J_{\sigma}(z_0)z_0)^{-k_1},\dots,
                               (\mathcal J_{\sigma}(z_0)z_0)^{-k_N}\Big).
\]
We now rewrite the determinant in the \emph{skew block} form:
\begin{align*}
  \mathfrak C_{\sigma}(\{k_1,\ldots,k_N\})
  &=\det\mathbf K_n
   \;=\;\sqrt{\det \begin{pmatrix} 0 & \mathbf K_n \\ -\mathbf K_n^{\top} & 0 \end{pmatrix}}\\
  &=\sqrt{\det
      \left(
      \begin{pmatrix} D_n & 0 \\ 0 & D_n^{-\top} \end{pmatrix}
      \begin{pmatrix} 0 & \mathbf K_n \\ -\mathbf K_n^{\top} & 0 \end{pmatrix}
      \begin{pmatrix} D_n^{-1} & 0 \\ 0 & D_n^{\top} \end{pmatrix}
      \right)}\\
  &=\sqrt{\det \begin{pmatrix}
       0 & D_n\,\mathbf K_n\,D_n^{-1} \\
      -D_n^{-\top}\mathbf K_n^{\top}D_n^{\top} & 0
     \end{pmatrix}} .
\end{align*}

Multiplying the off-diagonal blocks by the common scalar $(c_2 n^{1/3})$ multiplies
the determinant (inside the square root) by $(c_2 n^{1/3})^{2N}$. Hence,
\begin{equation}\label{eq:block-identity}
  (c_2 n^{1/3})^{N}\,
  \mathfrak C_{\sigma}(\{k_1,\ldots,k_N\})
  \;=\;
  \sqrt{\det \begin{pmatrix}
     0 & (c_2 n^{1/3})\,D_n\,\mathbf K_n\,D_n^{-1}\\
    -(c_2 n^{1/3})\,D_n^{-\top}\mathbf K_n^{\top}D_n^{\top} & 0
  \end{pmatrix}} .
\end{equation}
By the definition of the conjugated kernel
$\widetilde{\mathcal K}_{\sigma}(i,j)
=(\mathcal J_{\sigma}(z_0)z_0)^{-i}\mathcal K_{\sigma}(i,j)
 (\mathcal J_{\sigma}(z_0)z_0)^{\,j}$,
the \((i,j)\)-entry of the upper right block equals
\[
  (c_2 n^{1/3})\,
  \widetilde{\mathcal K}_{\sigma}(k_i,k_j)
  =(c_2 n^{1/3})\,
    \widetilde{\mathcal K}_{\sigma}\big(c_1 n+c_2 n^{1/3}x_i,\,
                                          c_1 n+c_2 n^{1/3}x_j\big) \,+\, o(1),
\]
where the $o(1)$ is uniform for $(x_1,\ldots,x_N)$ in compact sets.
By Lemma~\ref{lem:alpha-three-regimes}, entrywise on compact sets,
\[
  (c_2 n^{1/3})\,
  \widetilde{\mathcal K}_{\sigma}\big(c_1 n+c_2 n^{1/3}x_i,\,
                                        c_1 n+c_2 n^{1/3}x_j\big)
  \ \longrightarrow\ K_{\mathrm{Airy}}(x_i,x_j),
\]
while the two \textit{mixed} regimes of Lemma~\ref{lem:alpha-three-regimes} tend to $0$ and therefore do not appear
in the present block.
Hence, the $2N\times2N$ block matrix in \eqref{eq:block-identity} converges entrywise to
\[
  \begin{pmatrix}
    0 & K_{\mathrm{Airy}}(x_i,x_j)\\
    -K_{\mathrm{Airy}}(x_j,x_i) & 0
  \end{pmatrix}_{1\le i,j\le N},
\]
whose determinant equals $\det\big(K_{\mathrm{Airy}}(x_i,x_j)\big)^{2}$.
Taking square roots and using the continuity of the determinant, we obtain
\[
  \lim_{n\to\infty}(c_2 n^{1/3})^{N}\,
  \mathfrak C_{\sigma}(\{k_1,\ldots,k_N\})
  \;=\;
  \det\big[K_{\mathrm{Airy}}(x_i,x_j)\big]_{i,j=1}^{N},
\]
uniformly on compact sets of $(x_1,\ldots,x_N)\in\mathbb R^N$.
\end{proof}

\begin{theorem}\label{thm:gen-mat}
Let $\sigma=(m,n,\alpha)$ with $m/n\to\tau\in(0,\infty)$ and $0<\alpha<1$.
There exist positive constants $c_1=c_1(\alpha,\tau,t)$, $c_2=c_2(\alpha,\tau,t)$
given by \eqref{eq:doublesaddle}-\eqref{eq:c2-from-sigma} such that,
for every fixed $M\ge1$ and $a_1,\dots,a_M\in\mathbb R$,
\[
  \lim_{n\to\infty}
  \mathbb P_{\sigma}\Big(\,\lambda\ :\
     \frac{\lambda_i-c_1 n}{c_2 n^{1/3}}<a_i,\quad 1\le i\le M\,\Big)
  \;=\;
  \mathbb P_{\mathrm{Airy}}\big(\zeta_i<a_i,\ 1\le i\le M\big),
\]
i.e.\ the rescaled top rows converge jointly to the Airy ensemble,
and in particular
\[
  \frac{\lambda_1-c_1 n}{c_2 n^{1/3}}\ \Rightarrow\ \mathrm{TW}_{2}.
\]
\end{theorem}

\begin{proof}
Take $k_i=\lfloor c_1 n+c_2 n^{1/3}x_i\rfloor$ and invoke
Proposition~\ref{prop:alpha-det-limit} to get convergence of $k$-point functions.
The rest proof of this theorem is similar to that of Theorem~\ref{thm:pois-t-schur-airy}.
\end{proof}

Hence, Matsumoto's main result~\cite[Theorem 1]{Mat1} is recovered here, when choosing $M=1$ in Theorem~\ref{thm:gen-mat}.  Specialising $t=0$ gives us Johansson's result for the Schur measure, as presented in~\cite{Joh2}.

\subsection{Application to $t$-ascent pairs }
\label{subsec:t-ascent-limit}

Let $L^{(t)}$ be the longest $t$-ascent pair length in the random marked permutation
model. By Theorem~\ref{thm:t-ascent-equals-lambda1},
$L^{(t)}\stackrel{d}{=}\lambda_1$ under the fixed-size  $t$-Schur measure,
in the Poissonised case, $\lambda_1$ is the top particle of the determinantal process
with kernel $\mathcal K_t$.

\begin{corollary}\label{cor:t-ascent-limits}
In the settings of Theorems~\ref{thm:pois-t-schur-airy}, \ref{thm:depois-t}
and~\ref{thm:rect-airy}, the same limits hold for $L^{(t)}$:
\[
  \frac{L^{(t)}-2\sqrt{\kappa}}{\kappa^{1/6}}\Rightarrow \mathrm{TW}_{2}
  \quad\text{(Poissonised),}\qquad
  \frac{L^{(t)}-2\sqrt{N}}{N^{1/6}}\Rightarrow \mathrm{TW}_{2}
  \quad\text{(fixed size)}.
\]
\end{corollary}

\begin{remark}
The shifted Schur measure case is Pfaffian and, in its Plancherel regime, has centring $2\sqrt{2N}$ with scale $(2N)^{1/6}$. Our $t$-Schur case is determinantal: in the Poissonised regime, $t$ enters only through the effective parameter $\kappa=(1-t)\xi$ (or is absorbed by a balanced normalisation), while for fixed size the edge scaling is the Schur one: centre $2\sqrt{N}$ and scale $N^{1/6}$.
\end{remark}

\vskip 0.5in

\noindent{\bf Conflict of interest statement}. The authors have no conflicts of interest to declare.

\vskip 0.5in

\noindent{\bf Data availability}. All data of this work are included in the manuscript.

	\bigskip
	\centerline{\bf Acknowledgements}
G.G. was supported by the Singapore Ministry of Education Academic Research Fund; grant numbers: RG14/24 (Tier 1) and MOET2EP20222-0005 (Tier 2). N.J. is partially supported by Simons Foundation grant MP-TSM-00002518 and NSFC grant No. 12171303. H.Z. would like to thank the Research Scholarship awarded by NTU.

\bibliographystyle{amsalpha}

\begin{thebibliography}{99}

\bibitem{AD}
D. Aldous and P. Diaconis,
\textit{Longest increasing subsequences: From patience sorting to the Baik-Deift-Johansson theorem},
Bull. Amer. Math. Soc. \textbf{36} (1999), 413--432.

\bibitem{AvM}
M. Adler and P. van Moerbeke,
\textit{Integrals over classical groups, random permutations, Toda and Toeplitz lattices},
Comm. Pure Appl. Math. \textbf{54} (2001), 153--205.

\bibitem{BDJ}
J. Baik, P. Deift and K. Johansson,
\textit{On the distribution of the length of the longest increasing subsequence in a random permutation},
J. Amer. Math. Soc. \textbf{12} (1999), 1119--1178.

\bibitem{BR1}
J. Baik and E. M. Rains,
\textit{Symmetrized random permutations},
in \textit{Random Matrix Models and their Applications}, eds. P. Bleher and A. Its,
Math. Sci. Res. Inst. Publications \textbf{40}, Cambridge Univ. Press, 2001, pp. 1--19.

\bibitem{BR2}
J. Baik and E. R. Rains,
\textit{The asymptotics of monotone subsequences of involutions},
Duke Math. J. \textbf{109} (2001), 205--281.

\bibitem{BR3}
J. Baik and E. R. Rains,
\textit{Limiting distributions for a polynuclear growth model with external sources},
J. Statist. Phys. \textbf{100} (2000), 523--541.

\bibitem{BlO}
S.~Bloch and A.~Okounkov,
\textit{The character of the infinite wedge representation},
Adv. Math. \textbf{149} (2000), no.~1, 1-60.


\bibitem{BO}
A. Borodin and A. Okounkov,
\textit{A Fredholm determinant formula for Toeplitz determinants},
Int. Eqns. Oper. Th. \textbf{37} (2000), 386--396.



\bibitem{BOO}
A. Borodin, G. Olshanski and A. Okounkov,
\textit{Asymptotics of Plancherel measures for symmetric groups},
J. Amer. Math. Soc. \textbf{13} (2000), 481--515.



\bibitem{BS}
A. Böttcher and B. Silbermann,
\textit{Analysis of Toeplitz Operators},
Springer-Verlag, Berlin, 1990.

\bibitem{BW}
E. Basor and H. Widom,
\textit{On a Toeplitz determinant identity of Borodin and Okounkov},
Int. Eqns. Oper. Th. \textbf{37} (2000), 397--401.

\bibitem{BoPer}
A. Borodin,
\textit{Periodic Schur process and cylindric partitions},
Duke Math. J. \textbf{140} (2007), no.~3, 391-468.

\bibitem{BC}
A. Borodin and I. Corwin,
\textit{Macdonald processes},
Probab. Theory Relat. Fields, \textbf{158} (2014), 225-400.

\bibitem{BufPet}
A. I. Bufetov and L. Petrov,
\textit{Yang-Baxter field for spin Hall-Littlewood symmetric functions},
Forum Math. Sigma \textbf{7} (2019), e39.


\bibitem{CJ} W. Cai, N. Jing, {\em On vertex operator realizations of Jack functions}, J. Algebra Comb. {\bf 32} (2010), 579--595.

\bibitem{DJKM} E. Date, M. Jimbo, M. Kashiwara and T. Miwa, {\it Transformation groups for
soliton equations, nonlinear integrable systems-classical theory and quantum
theory} pp. 39-l 19, Kyoto, World Scientific, Singapore, 1983.

\bibitem{FJK} B. Fauser, P.~D. Jarvis and R.~C. King, {\it Plethysms, replicated Schur functions and
series, with applications to vertex operators}, J. Phys. A.: Math. Theor. 43 (2010) 405202
(30pp).

\bibitem{FLM} I. Frenkel, J. Lepowsky and A. Meurman, {\em Vertex operator algebras and the Monster}, Academic Press, New York, 1988.

\bibitem{FPS}
I.~Frenkel, I.~Penkov and V.~Serganova,
\textit{A categorification of the boson--fermion correspondence via representation theory of $sl(\infty)$},
Commun.\ Math.\ Phys.\ \textbf{341} (2016), 911--931.


\bibitem{FH} W. Fulton and J. Harris, {\it Representation Theory}, Springer-Verlag, New York, 1991.


\bibitem{IIKS}
A.~R.~Its, A.~G.~Izergin, V.~E.~Korepin, and N.~A.~Slavnov,
\textit{Differential equations for quantum correlation functions},
Int. J. Mod. Phys. B \textbf{4} (1990), no.~5, 1003-1037.

\bibitem{Ge}
I.~M.~Gessel,
\textit{Symmetric functions and P-recursiveness},
J. Combin. Theory Ser.~A \textbf{53} (1990), no.~2, 257-285.



\bibitem{J1} N. Jing, {\em Vertex operators, symmetric functions and the spin groups $\Gamma_n$}, J. Algebra {\bf 138} (1991), 340-398

\bibitem{J2} N. Jing, {\em Vertex operators and Hall-Littlewood symmetric functions}, Adv. Math. {\bf 87} (1991), 226-248.

\bibitem{J3} N. Jing, {\em $q$-hypergeometric series and Macdonald functions}, J. Algebr. Comb. {\bf 3} (1994), 291--305.

\bibitem{J4} N. Jing, T. J\'ozefiak, {\em  A formula for two row Macdonald functions}, Duke Math. J. {\bf 67} (1992), no. 2, 377--385.

\bibitem{Joh}
K.~Johansson,
\textit{Discrete orthogonal polynomial ensembles and the Plancherel measure},
Annals Math. \textbf{153} (2001), no.~1, 259-296.

\bibitem{Joh2}
K.~Johansson,
\textit{Shape fluctuations and random matrices},
Comm. Math. Phys. \textbf{209} (2000), no.~2, 437-476.

\bibitem{Kac}
V.~G.~Kac,
\textit{Infinite-dimensional Lie algebras},
3rd ed., Cambridge University Press, Cambridge, 1990.


\bibitem{KO}
S.~Kerov and G.~Olshanski, \textit{Polynomial functions on the set of Young diagrams},
C. R. Acad. Sci. Paris S\'er.~I Math. \textbf{319} (1994), no.~2, 121-126.

\bibitem{KOV}
S.~V.~Kerov, G.~I.~Olshanski, and A.~M.~Vershik,
\textit{Harmonic analysis on the infinite symmetric group. A deformation of the regular representation},
C. R. Acad. Sci. Paris S\'er.~I Math. \textbf{316} (1993), no.~8, 773-778.


\bibitem{K} R.~C. King, {\em S-functions and characters of Lie algebras and superalgebras}. Invariant theory and tableaux (Minneapolis, MN, 1988), pp.226-261,
IMA Vol. Math. Appl., 19, Springer, New York, 1990.

\bibitem{Li} D.~E. Littlewood, {\em The theory of group characters and matrix representations of
groups}, 2nd ed. Oxford University Press, London, 1950.

\bibitem{Ma} I.~G. Macdonald, {\em Symmetric functions and Hall polynomials}, 2nd ed., Clarendon Press,
Oxford, 1995.

\bibitem{Mat1}
S. Matsumoto,
{\em A scaling limit for $t$-Schur measures},
Kyushu J. Math. \textbf{59}, (2005) No. 1, 25-38 (2005).

\bibitem{Mat2}
S. Matsumoto,
{\em Correlation functions of the shifted Schur measure},
J. Math. Soc. Japan \textbf{57}, (2005) No. 3, 619-637.

\bibitem{O}
A. Okounkov,
{\em Infinite wedge and random partitions},
Sel. Math., New Ser. \textbf{7}, (2001) No. 1, 57-81.

\bibitem{Sch}
C.~Schensted,
\textit{Longest increasing and decreasing subsequences},
Canad. J. Math. \textbf{13} (1961), no.~2, 179-191.

\bibitem{TW}
C. Tracy, H. Widom,
\textit{A limit theorem for shifted Schur measures},
Duke Math. J. \textbf{123}, (2004), No. 1, 171-208.

\bibitem{W} H. Weyl, {\em The classical groups; their invariants and representations},
Princeton Univ. Press, Princeton, 1946.

\bibitem{Za} M. Zabrocki, {\em A Macdonald vertex operator and standard tableaux statistics for the two-column $(q, t)$-Kostka coefficients}, Electron. J. Comb. {\bf 5}, (1998) 45.

\bibitem{Ze} A. Zelevinsky, {\em Representations of finite classical groups}, Springer-Verlag, New York, 1981.

\end{thebibliography}

\end{document}